\documentclass[letterpaper, 10pt, journal, twoside]{IEEEtran}
\IEEEoverridecommandlockouts

\usepackage{cite}

\usepackage{amsthm}
\usepackage{amsmath}
\usepackage{amssymb}
\usepackage{amsfonts}
\usepackage{bm}
\let\labelindent\relax
\usepackage{enumitem}
\usepackage{algorithm}
\usepackage{algpseudocode}
\usepackage{graphicx}
\usepackage[mathscr]{eucal}
\usepackage{xcolor}
\usepackage[hidelinks]{hyperref}

\graphicspath{{figures/}}
\definecolor{gem-1}{rgb}{0.066,0.443,0.745}
\definecolor{gem-2}{rgb}{0.866,0.329,0}
\definecolor{gem-3}{rgb}{0.929,0.694,0.125}
\definecolor{gem-4}{rgb}{0.521,0.086,0.819}
\definecolor{gem-5}{rgb}{0.231,0.666,0.196}
\definecolor{gem-6}{rgb}{0.184,0.745,0.937}
\definecolor{gem-7}{rgb}{0.819,0.015,0.545}
\definecolor{gem-8}{rgb}{1,0.839,0.039}
\definecolor{gem-9}{rgb}{0.396,0.509,0.992}
\definecolor{gem-10}{rgb}{1,0.27,0.227}
\definecolor{gem-11}{rgb}{0,0.639,0.639}
\definecolor{gem-12}{rgb}{0.796,0.517,0.364}
\definecolor{onehalf-0}{HTML}{383a42} 
\definecolor{onehalf-1}{HTML}{e45649} 
\definecolor{onehalf-2}{HTML}{50a14f} 
\definecolor{onehalf-3}{HTML}{c18401} 
\definecolor{onehalf-4}{HTML}{0184bc} 
\definecolor{onehalf-5}{HTML}{a626a4} 
\definecolor{onehalf-6}{HTML}{0997b3} 
\definecolor{onehalf-7}{HTML}{fafafa} 

\providecommand{\changed}[1]{}
\renewcommand{\changed}[1]{\textcolor{onehalf-3}{#1}}

\newcommand{\mat}[1]{\begin{bmatrix} #1 \end{bmatrix}}

\newtheorem{theorem}{Theorem}
\newtheorem{definition}{Definition}
\newtheorem{lemma}{Lemma}
\newtheorem{proposition}{Proposition}
\newtheorem{assumption}{Assumption}
\newtheorem{corollary}{Corollary}
\newtheorem{remark}{Remark}

\newtheorem{property}{Property}
\newtheorem{problem}{Problem}

\DeclareMathOperator{\argmin}{argmin}

\newcommand{\N}{\mathbb{N}}

\newcommand{\R}{\mathbb{R}}

\newcommand{\cK}{\mathscr{K}}
\newcommand{\cL}{\mathscr{L}}

\newcommand{\sC}{\mathcal{C}}
\newcommand{\sD}{\mathcal{D}}

\newcommand{\sI}{\mathcal{I}}

\newcommand{\sK}{\mathcal{K}}

\newcommand{\sN}{\mathcal{N}}

\newcommand{\sR}{\mathcal{R}}

\newcommand{\sT}{\mathcal{T}}
\newcommand{\sU}{\mathcal{U}}
\newcommand{\sV}{\mathcal{V}}
\newcommand{\sW}{\mathcal{W}}
\newcommand{\sX}{\mathcal{X}}

\usepackage{siunitx}
\usepackage[normalem]{ulem}

\begin{document}
\title{Robust Dynamic Safety Margins are\\ Robust Control Barrier Functions}

\author{Victor Freire, Saber Jafarpour, and Marco M. Nicotra \thanks{This research is supported by the NSF-CMMI Award \#2411667. \emph{(Corresponding author: Victor Freire.)}}
\thanks{The authors are with the Department of Electrical, Computer \& Energy Engineering, University of Co\-lorado,  Boulder, CO 80309 USA (email: \{vifr9883, saber.jafarpour, marco.nicotra\}@colorado.edu).}}

\maketitle

\begin{abstract}
  This paper establishes robustness properties for control barrier functions (CBFs) constructed using dynamic safety margins (DSMs). We show that DSMs can act as robust CBFs by leveraging input-to-state stability for Lyapunov-based constructions, and contractivity for trajectory-prediction-based approaches. Additionally, we derive closed-form DSM-based robust CBFs (RDSM-RCBFs) specifically for linear systems. Simulations highlight the improved robustness of the proposed RDSM-RCBFs compared to nominal DSM-CBF baselines.
\end{abstract}

\section{Introduction}
Control barrier functions (CBFs) can be used to design safe controllers by filtering potentially unsafe inputs. They are powerful tools similar in nature to control Lyapunov functions, but with respect to set invariance instead of stability. Robust CBFs (RCBFs) can certify safety even when the system is subject to  unknown external disturbances. One of their main drawbacks is the lack of constructive methods for RCBFs.

\subsection{Related Literature}
To overcome the challenge of constructing CBFs, the literature has explored several frameworks. Hamilton--Jacobi (HJ) reachability analysis and sum-of-squares (SOS) optimization provide principled, systematic approaches for synthesizing control invariant sets and CBFs \cite{choi2021robust, wabersich2023data,wang2018permissive,clark2025semialgebraic}. However, HJ reachability suffers from the curse of dimensionality, restricting its application to low-dimensional systems, while SOS techniques are generally limited to polynomial dynamics. Learning-based techniques have also found success in practice for more general systems \cite{srinivasan2020synthesis, robey2020learning, bajelani2025data}. However, they usually lack formal safety guarantees and rely on accepting mild violations. Backup CBFs are an alternative constructive paradigm that is gaining popularity \cite{gurriet2020scalable, chen2021backup, van2024disturbance, van2026generalizations, dacs2026robust}. Backup CBFs rely on a safe backup controller to enlarge its associated invariant set using trajectory predictions. The approach has shown good performance, but suffers from high computational demand due to the trajectory predictions and usually results in an implicit form for the CBF. In \cite{van2026safety}, a closed-form expression was developed by interpolating between the safe backup controller and the unsafe nominal controller. Backup CBFs are closely related to the trajectory dynamic safety margin (DSM) approach proposed in \cite{freire2026designing}, which is extended to systems subjects to external disturbances in this paper. Dynamic safety margins are a component of the explicit reference governor (ERG) \cite{nicotra2018explicit} offering a measure of trajectory clearance to constraint boundaries under a prestabilizing feedback law. It was shown in \cite{freire2026using} that DSMs are CBFs over an augmented state-reference space. Other works relating the RG and CBF literatures include \cite{freire2023systematic, li2023governor, kaviani2026synthesis}

In parallel to the systematic design efforts, researchers worked on robustifying CBFs against model inaccuracies and external disturbances. Two distinct philosophies emerged to address this problem: On the one hand, works like \cite{jankovic2018robust} enforce worst-case disturbance rejection over bounded uncertainty sets. This is also the philosophy we take on this paper. On the other hand, \cite{kolathaya2019input} introduced input-to-state safety (ISSf) as a counterpart to input-to-state stability to generalize set invariance by guaranteeing that bounded disturbances result in bounded, predictable expansions of the invariant set. ISSf can generally achieve less conservative results but tends to allow small violations of the safety set which may be unacceptable in safety-critical applications. Extensions like tunable ISSf \cite{alan2021safe} further enhance robustness by dynamically adapting parameter estimates or disturbance bounds.

\subsection{Contributions}
The contributions of this paper are:
\begin{itemize}
  \item We formalize the definition of robust dynamic safety 
  margins (RDSMs) and show that they are RCBFs for an augmented system obtained by concatenating the state of the system and the reference to the prestabilizing controller that was used to construct the RDSM.
  \item We formulate a robust DSM-CBF safety filter in the augmented space that yields robust safe inputs. We also highlight that the resulting optimization problem is persistently feasible.
  \item We show how to construct RDSMs using two distinct approaches: One based on ISS Lyapunov functions and another based on trajectory predictions of the prestabilized system and its contraction properties.
  \item We demonstrate the performance of RDSM-RCBFs and their ability to reject disturbances in simulation examples. 
\end{itemize}

The rest of the paper is organized as follows: Section~\ref{sec:prelim} presents the definition of robust control invariance, its connection to robust control barrier functions, and formulates the RCBF safety filter that generates robust safe inputs; Section~\ref{sec:rdsm} defines the robust dynamic safety margin, shows RDSMs are RCBFs, and formulates the RDSM-RCBF safety filter in the augmented state-reference space; Section~\ref{sec:lyap-rdsm} shows how to construct RDSMs using ISS Lyapunov functions; Section~\ref{sec:traj-rdsm} shows how to construct RDSMs using trajectory predictions; Section~\ref{sec:examples} presents the numerical simulations; Section~\ref{sec:rdsm-rcbf-linear} shows how both of these RDSM design approaches simplify for linear systems subject to polyhedral constraints; Section~\ref{sec:conclusion} concludes the paper.

\emph{Notation:} Given a symmetric, positive definite matrix $P \in \R^{n \times n}$, the $P$-weighted norm is $\|\bm{x}\|_P = \|P^{1/2} \bm{x} \| = \sqrt{\bm{x}^\top P \bm{x}}$. With this, the weighted induced matrix norm is $\|A\|_P = \|P^{1/2} A P^{-1/2}\|$ and the weighted logarithmic norm is $\mu_P(A) = \lambda_{\max}(PAP^{-1} + A^\top)$, where $\lambda_{\max}(A)$ returns the largest eigenvalue of the symmetric part of $A$. 

\section{Preliminaries} \label{sec:prelim}
Consider a control- and disturbance-affine model
\begin{equation} \label{eq:sys-disturbance}
  \dot{\bm{x}} = f(\bm{x}) + g(\bm{x}) \bm{u} + p(\bm{x}) \bm{d},
\end{equation}
where $f:\R^n \to \R^n$, $g: \R^n \to \R^{n\times m}$ and $p:\R^n \to \R^{n\times n_d}$ are continuous functions. The system is subject to two types of inputs: $\bm{u} \in \R^m$ is the control action and $\bm{d} \in \R^{n_d}$ is an unknown external disturbance. The system's actuation capabilities are described by a convex, compact set denoted $\sU \subset \R^m$, and we call controls $\bm{u} \in \sU$ \emph{feasible}. Additionally, we limit our study of \eqref{eq:sys-disturbance} to bounded disturbances.
\begin{assumption} \label{ass:compact-convex-disturbance-set}
  The disturbance set $\sW \subset \R^{n_d}$ is compact, convex and $0 \in \sW$.
\end{assumption}
Let $\sX \subset \R^n$ be the state constraint set describing the safe operating region of system \eqref{eq:sys-disturbance} and assume it is closed. An obstacle, for instance, may not intersect $\sX$. We call states $\bm{x} \in \sX$ and sets $\sC \subset \sX$ \emph{admissible}. Along with the broad constrained control literature, this paper addresses the following problem.

\begin{problem}[Robust constrained control] \label{prob:robust-cst-ctrl}
  Given an ad\-mi\-s\-si\-ble initial state $\bm{x}_0 \in \sX$, find a feasible
  control signal $u:[0,\infty) \to \sU$ such that, in spite of any disturbance
  signal $d:[0,\infty) \to \sW$, the trajectory $x(t) \in \sX$ for all $t \geq
  0$; where $x:[0,\infty) \to \R^n$ is the unique solution to the initial
  value problem (IVP)
  \begin{equation}\label{eq:ivp-disturbance}
    \dot{x}(t) \!=\! f\big(x(t)\big) + g\big(x(t)\big) u(t) + p\big(x(t)\big)
    d(t), \; x(0) \!=\! \bm{x}_0.
  \end{equation}
\end{problem}
We call system \eqref{eq:sys-disturbance} \emph{robustly safe} if it admits solutions to Problem~\ref{prob:robust-cst-ctrl} for any $\bm{x}_0 \in \sX$.

\subsection{Robust Control Invariance}
While robust safety is a difficult property to verify, the robust control invariance (RCI) property is closely related.
\begin{definition}\label{def:rci}
  A closed set $\sC \subset \R^n$ is \emph{robust control invariant} (RCI) if for any initial state $ \bm{x}_0 \in \sC$, there exists a horizon $T>0$, and an input signal $u:[0,T] \to \sU$ for which, in spite of  any disturbance signal $d:[0,T] \to \sW$, at least one solution to the initial value problem \eqref{eq:sys-disturbance}, with $x(0) = \bm{x}_0$, remains in $\sC$. That is, $\forall t \in [0,T]$, $x(t) \in \sC$.
\end{definition}
Definition~\ref{def:rci} is usually called local viability \cite{aubin1991viability} or weak control invariance \cite{blanchini1999set} and differs from the usual RCI definition found in the literature in that the solution to the IVP \eqref{eq:ivp-disturbance} may not be unique. While uniqueness is a desirable property, the Lipschitz continuity assumptions it requires are commonly found in the literature and studied in other works. In this paper, we focus on providing a faithful representation of the viability theorems found in \cite{aubin1991viability}. Regarding the finite horizon $T$, note that only two options are possible because $\sC$ is closed \cite[Theorem 3.3.4]{aubin1991viability}: Either $T=\infty$, or the solution exhibits finite escape at $T<\infty$ (which is safe behavior provided $\sC \subset \sX$, but will never happen if $\sC$ is bounded \cite[Theorem 3.3.5]{aubin1991viability}).

The Nagumo theorem provides a characterization of RCI in terms of the Bouligand contingent cone \cite[Definition~4.1.1]{aubin2009set}.
\begin{theorem}[Nagumo \cite{nagumo1942lage}]\label{thm:rnagumo} 
  A closed set $\sC \subset \R^n$ is RCI if and only if $\forall \bm{x} \in \sC$, $\exists \bm{u} \in \sU$ such that $\forall \bm{d} \in \sW$, 
  \begin{equation} \label{eq:rnagumo}
    f(\bm{x}) + g(\bm{x}) \bm{u} + p(\bm{x}) \bm{d} \in \sT_{\sC}(\bm{x}),
  \end{equation}
  where $\sT_{\sC}(\bm{x})$ is the contingent cone to $\sC$ in $\bm{x}$.
\end{theorem}

While the Nagumo theorem characterization was groundbreaking, working with the contingent cone can be cumbersome. For sets defined as the superlevel of continuously differentiable functions
\begin{equation}\label{eq:set-superlevel}
    \sC = \left\{\bm{x} \in \R^n \;\Big|\; \min_{i \in \{1,\ldots,q\}} h_i(\bm{x}) \geq 0\right\},
\end{equation}
the contingent cone simplifies greatly provided they satisfy a transversality condition captured by the following constraint qualification.
\begin{property}[MFCQ \cite{mangasarian1967fritz}]
  Let $h: \R^n \to \R^{q}$ be a continuously differentiable function with superlevel set $\sC$ given in \eqref{eq:set-superlevel}. We say the \emph{MFCQ} holds if  \begin{equation}\label{eq:MFCQ}
    \forall \bm{x} \in \sC, \quad \exists \bm{z} \in \R^n, \quad \forall i \in \sI(\bm{x}), \quad D h_i(\bm{x}) \bm{z} > 0,      
  \end{equation}
  where $\sI(\bm{x}) = \{i \mid h_i(\bm{x}) = 0\}$ is the set of active constraints.
\end{property}
Sets $\sC$ defined as in \eqref{eq:set-superlevel} for which the MFCQ holds were called \emph{practical} in \cite{blanchini2008set}.

\begin{proposition} \label{prop:contingent-cone-practical}
  Let $h:\R^n \to \R^q$ be a continuously differentiable function and let $\sC \subset \R^n$ be its superlevel set \eqref{eq:set-superlevel}. The following hold:
  \begin{enumerate}
    \item General case:
      \begin{equation*}
        \forall \bm{x} \in \partial \sC, \; \sT_{\sC}(\bm{x}) \subset
        \left\{\bm{z} \in \R^n \;\Big|\; \min_{i \in \sI(\bm{x})} Dh_i(\bm{x})
        \bm{z} \geq 0 \right\}.
      \end{equation*}
    \item Practical set: If the MFCQ holds, then
      \begin{equation*}
        \forall \bm{x} \in \partial \sC, \; \sT_{\sC}(\bm{x}) =
        \left\{\bm{z} \in \R^n \;\Big|\; \min_{i \in \sI(\bm{x})} Dh_i(\bm{x})
        \bm{z} \geq 0 \right\}.
      \end{equation*}
  \end{enumerate}
\end{proposition}
Given this simplification, the following result is the immediate specialization of Theorem \ref{thm:rnagumo} to practical sets.
\begin{theorem} \label{thm:rnagumo-practical}
  Let $h: \R^n \to \R^q$ be continuously differentiable with superlevel set $\sC$ given in \eqref{eq:set-superlevel}. If $\sC$ is RCI, then
  \begin{equation} \label{eq:rnagumo-practical}
    \forall \bm{x} \in \partial \sC, \; \exists \bm{u} \in \sU, \; \forall \bm{d} \in \sW, \; \min_{i \in \sI(\bm{x})} \dot{h}_i(\bm{x}, \bm{u}, \bm{d}) \geq 0.
  \end{equation}
  Conversely, if \eqref{eq:rnagumo-practical} and the MFCQ hold, then $\sC$ is RCI. 
\end{theorem}
Note that
\begin{equation}
  \dot{h}_i(\bm{x}, \bm{u}, \bm{d}) = L_fh_i(\bm{x}) + L_gh_i(\bm{x})\bm{u} + L_ph_i(\bm{x}) \bm{d},
\end{equation}
where $L_fh_i$, $L_gh_i$ and $L_ph_i$ are the Lie derivatives of the $i$-th channel of $h$ along $f$, $g$ and $p$, respectively.

\subsection{Robust Control Barrier Functions}
While condition \eqref{eq:rnagumo-practical} is significantly easier to verify than the contingent cone condition \eqref{eq:rnagumo}, it applies only at the boundary of the set $\partial \sC$, which has measure zero. This can cause issues for numerical methods and limits the practical applicability of Theorem \ref{thm:rnagumo-practical}. Robust control barrier functions provide a certificate of robust control invariance and regularize condition \eqref{eq:rnagumo-practical} to the entire set $\sC$ beyond its boundary by introducing a class $\cK$ term.

\begin{definition}\label{def:rcbf}
  A continuously differentiable function $h: \R^n \to \R^{q}$ with superlevel set $\sC$ given in \eqref{eq:set-superlevel} is a \emph{robust control barrier function} (RCBF) if there exists $\alpha \in \cK$ satisfying $\forall \bm{x} \in \sC$, $\exists \bm{u} \in \sU$, such that
  \begin{equation} \label{eq:rcbf-cond}
    \forall \bm{d} \in \sW, \quad \min_{i \in \{1,\ldots,q\}} \left[\dot{h}_i(\bm{x}, \bm{u}, \bm{d}) + \alpha\big(h_i(\bm{x})\big)\right] \geq 0.
  \end{equation}
\end{definition}

Let us define the \emph{robust regulation map} of $\sC$ by
\begin{equation} \label{eq:rregulation-map}
    \sK_{\sC}(\bm{x}) = \left\{ \bm{u} \in \sU \mid \forall \bm{d} \in \sW, \; \dot{\bm{x}} \in \sT_{\sC}(\bm{x}) \right\},
\end{equation}
where $\dot{\bm{x}}$ is given in \eqref{eq:sys-disturbance}.
Robust control barrier functions certify the robust control invariance of their superlevel set. Under additional conditions, any continuously differentiable function $h:\R^n \to \R^q$ describing a robust control invariant set $\sC$ is an RCBF.
\begin{theorem} \label{thm:rcbf-rci}
  Let $h:\R^n \to \R^q$ be a continuously differentiable function with superlevel set $\sC$ given in \eqref{eq:set-superlevel}. The following hold:
  \begin{enumerate}
    \item If $h$ is an RCBF and the MFCQ holds, then $\sC$ is RCI. \label{thm:rcbf->rci}
    \item If $\sC$ is RCI, compact, and there exists a continuous selection of $\sK_{\sC}$, then $h$ is an RCBF. \label{thm:rci->rcbf}
  \end{enumerate}
\end{theorem}


\begin{proof}
  We begin by showing statement \ref{thm:rcbf->rci}. Let $\alpha \in \cK$ satisfy the RCBF condition \eqref{eq:rcbf-cond} at $\bm{x} \in \partial \sC$ with $\bm{u} \in \sU$. Pick $i \in \sI(\bm{x})$ and note that $\forall \bm{d} \in \sW$,
  \begin{align*}
    0 &\leq L_fh_i(\bm{x}) + L_gh_i(\bm{x}) \bm{u} + L_ph_i(\bm{x}) \bm{d} + \alpha\big(h_i(\bm{x})\big) \\
    &\leq L_fh_i(\bm{x}) + L_gh_i(\bm{x}) \bm{u} + L_ph_i(\bm{x}) \bm{d} + \alpha(0) \\
    &\leq L_fh_i(\bm{x}) + L_gh_i(\bm{x}) \bm{u} + L_ph_i(\bm{x}) \bm{d}
  \end{align*}
  We conclude by Theorem \ref{thm:rnagumo-practical} that $\sC$ is robust control invariant.

  Next, we show statement \ref{thm:rci->rcbf}. Let $\mu:\R^n \to \R^m$ be a continuous selection of the robust regulation map $\sK_{\sC}$. Pick any $i \in \{1,\ldots,q\}$ and define the set-valued map $\hat{\sC}_i:[0,\infty) \rightrightarrows \R^n$ by
  \begin{equation}
    \hat{\sC}_i(c) = \{\bm{x} \in \sC \mid 0 \leq h_i(\bm{x}) \leq c\},
  \end{equation}
  the function $\xi_i: \R^n \times \R^{n_d} \to \R$ by
  \begin{equation}
      \xi_i(\bm{x}, \bm{d}) = L_fh_i(\bm{x}) + L_gh_i(\bm{x}) \mu(\bm{x}) + L_ph_i(\bm{x}) \bm{d},
  \end{equation}
  and the value function $\hat{\alpha}_i: [0, \infty) \to \R$ by
  \begin{equation}
    \hat{\alpha}_i(c) = - \inf_{(\bm{x}, \bm{d}) \in \hat{\sC}_i(c) \times \sW} \xi_i(\bm{x}, \bm{d}).
  \end{equation}
  We seek to upper-bound $\hat{\alpha}_i$ by a class $\cK$ function. Leveraging \cite[Lemma 4.1]{albertini1999continuous}, it can be shown that such an upper bound exists if $\hat{\alpha}_i$ is upper semicontinuous and $\hat{\alpha}(0) \leq 0$.

  For the first condition ($\hat{\alpha}_i$ upper semicontinuous), we will invoke the Berge maximum theorem \cite{berge1963topological}. We begin by showing that the set-valued map $c \mapsto \hat{\sC}_i(c) \times \sW$ is upper semicontinuous: Let $\{c_k\} \subset  [0,\infty)$ be a sequence that converges to $c \in [0,\infty)$ and let $\{(\bm{x}_k, \bm{d}_k)\} \subset \R^n \times \R^{n_d}$ be a sequence that converges to $(\bm{x}, \bm{d}) \in \R^n \times \R^{n_d}$ satisfying $(\bm{x}_k, \bm{d}_k) \in \hat{\sC}_i(c_k) \times \sW$. Since $\{(\bm{x}_k, \bm{d}_k)\} \subset \sC \times \sW$ and $\sC \times \sW$ is closed, we have $(\bm{x}, \bm{d}) \in \sC \times \sW$.  Noting that for all $k \in \N$, $0 \leq h_i(\bm{x}_k) \leq c_k$, it follows by continuity of $h_i$ that $0 \leq h_i(\bm{x}) \leq c$, and thus $\bm{x} \in \hat{\sC}_i(c)$. Sinc $\hat{\sC}_i(c) \times \sW \subset \sC \times \sW$ is compact, we conclude by the sequential characterization of upper semicontinuity for compact set-valued maps that $c \mapsto \hat{\sC}_i(c) \times \sW$ is upper semicontinuous at $c$. Since $c \in [0,\infty)$ was arbitrary, $c \mapsto \hat{\sC}_i(c) \times \sW$ is upper semicontinuous.  Recalling that $c \mapsto \hat{\sC}_i(c) \times \sW$ has compact values, it follows by the Berge maximum theorem \cite{berge1963topological} that the value function $\hat{\alpha}_i$ is upper semicontinuous.

  For the second condition ($\hat{\alpha}(0) \leq 0$), pick any $(\bm{x}, \bm{d}) \in \hat{\sC}_i(0) \times \sW \subset \sC \times \sW$ and note that since $\mu(\bm{x}) \in \sK_{\sC}(\bm{x})$, $f(\bm{x}) + g(\bm{x}) \mu(\bm{x}) + p(\bm{x}) \bm{d} \in \sT_{\sC}(\bm{x})$. It follows by Proposition \ref{prop:contingent-cone-practical} that $\xi_i(\bm{x}, \bm{d}) \geq 0$, and thus $\hat{\alpha}_i(0) \leq 0$.

  Having shown its conditions hold, it follows by \cite[Lemma 4.1]{albertini1999continuous} that there exists $\alpha_i \in \cK$ that upper-bounds $\hat{\alpha}_i$. Since $i \in \{1,\ldots, q\}$ was arbitrary, let us define $\alpha(c) = \max_{i \in \{1,\ldots,q\}} \alpha_i(c)$ and note that $\alpha \in \cK$.
  
  Finally, let $\bm{x} \in \sC$
  be given and recall $\mu(\bm{x}) \in \sU$. For any $i \in \{1,\ldots,q\}$
  and any $\bm{d} \in \sW$,
  \begin{align*}
    \xi_i(\bm{x}, \bm{d}) &= L_fh_i(\bm{x}) + L_gh_i(\bm{x}) \mu(\bm{x}) + L_ph_i(\bm{x}) \bm{d}\\
    &\geq  \inf_{(\bm{x}_0, \bm{d}_0) \in
    \hat{\sC}_i\big(h_i(\bm{x})\big) \times \sW} \xi_i(\bm{x}_0, \bm{d}_0) \\
    &\geq -\hat{\alpha}_i\big( h_i(\bm{x})\big) \\
    &\geq -\alpha_i\big( h_i(\bm{x}) \big) \\
    &\geq -\alpha\big(h_i(\bm{x})\big).
  \end{align*}
  We conclude that $h$ is an RCBF.
\end{proof}

\begin{remark} \label{rem:ecs}
  All the presented results can be extended to families of vector-valued functions $h_{\tau}:\R^n \to \R^q$ as long as the index $\tau \in [0,1]$ belongs to a compact set and the joint map $(\tau, \bm{x}) \mapsto h_{\tau}(\bm{x})$ is continuously differentiable. This is particularly useful when working with backup RCBFs \cite{van2024disturbance} and trajectory RDSMs as we will see in upcoming sections. In these cases, we say the family $\{h_{\tau}\}$ has the \emph{extended robust control-sharing property}. The basic control-sharing property was introduced in \cite{xu2018constrained} for CBFs and generalized in \cite{freire2026using} to continuously differentiable functions.
\end{remark}

Besides providing the RCI certificate, RCBFs also allow the construction of safe optimization-based controllers with RCBF safety filters.

\subsection{RCBF Safety Filters}
Given an RCBF $h: \R^n \to \R^q$, the RCBF safety filter is
\begin{align} \label{eq:rcbf-sf}
  \min_{\bm{u}\in \sU} \quad & \|\bm{u} - \kappa(\bm{x})\|^2\\
  \text{s.t.} \quad & \dot{h}_i(\bm{x}, \bm{u}, \bm{d}) \geq -\alpha\big(h_i(\bm{x})\big), \quad \forall \bm{d} \in \sW, \; i = 1,\ldots,q. \nonumber
\end{align}
This optimization problem is intractable because it has infinitely many constraints ($\forall \bm{d} \in \sW$). To address this, we can leverage compactness and convexity of $\sW$ to state an equivalent worst-case safety filter
\begin{align} \label{eq:rcbf-sf-worst-case}
  \min_{\bm{u}\in \sU} \quad & \|\bm{u} - \kappa(\bm{x})\|^2\\
  \text{s.t.} \quad & L_fh_i(\bm{x}) + L_gh_i(\bm{x})\bm{u} + \omega_i(\bm{x}) \geq -\alpha\big(h_i(\bm{x})\big), \nonumber \\ &  i = 1,\ldots,q, \nonumber
\end{align}
where
\begin{equation}
  \omega_i(\bm{x}) = \min_{\bm{d}\in \sW} L_ph_i(\bm{x}) \bm{d}.
\end{equation}
Note that computing $\omega_i(\bm{x})$ is easy because $\sW$ is convex by Assumption \ref{ass:compact-convex-disturbance-set}. Moreover, if $\sW$ is a polytope, $\omega_i$ involves solving a linear program (LP) and its solution can usually be found in closed-form if $\sW$ is an interval. Irrespective of the geometry of $\sW$, if $\sU$ is a polytope, then \eqref{eq:rcbf-sf} is a quadratic program (QP) and can be solved efficiently.

\section{Robust Dynamic Safety Margins} \label{sec:rdsm}
In this paper, we leverage tools from the reference governor (RG) literature \cite{garone2017reference} to design RCBFs. This section introduces the standard assumptions in the RG literature and formulates robust dynamic safety margins (RDSMs). Then, we show their connections to RCBFs.
\begin{assumption} \label{ass:iss}
  In the absence of disturbances (e.g., $d(t) \equiv 0$), system \eqref{eq:sys-disturbance} admits an equilibrium manifold parameterized by a continuous injection $\bm{v} \mapsto \big(\bar{x}(\bm{v}), \bar{u}(\bm{v})\big)$ satisfying
  \begin{equation}
    \forall \bm{v} \in \R^{\ell}, \quad f\big(\bar{x}(\bm{v})\big) + g\big(\bar{x}(\bm{v})\big)
    \bar{u}(\bm{v}) = 0,
  \end{equation}
  Moreover, every equilibrium point $\bar{x}(\bm{v})$ is continuously input-to-state stabilizable (ISS) with respect to the disturbance $\bm{d}$.
\end{assumption}
Assumption 1 is standard in the reference governor literature: it applies to any system that admits a connected path of stabilizable equilibrium points \cite{garone2017reference}; this encompasses a wide range of meaningful systems and applications. In fact, any linear system (in the form of \eqref{eq:sys-disturbance-linear}) with a stabilizable pair $(A,B)$ satisfies Assumption~\ref{ass:iss}.

The vector $\bm{v} \in \R^{\ell}$ is hereafter called a \emph{reference} for the system because it identifies a specific equilibrium point. By Assumption \ref{ass:iss}, let $\pi: \R^n \times \R^{\ell} \to \R^m$ be a continuous function such that $\pi\big(\bar{x}(\bm{v}), \bm{v}\big) = \bar{u}(\bm{v})$, and $\bar{x}(\bm{v})$ is an ISS equilibrium point of the \emph{prestabilized system} subject to disturbances
\begin{equation}\label{eq:sys-disturbance-prestab}
  \dot{\bm{x}} = f_{\pi}(\bm{x}, \bm{v}) + p(\bm{x}) \bm{d},
\end{equation}
with respect to $\bm{d}$, where $f_{\pi}(\bm{x}, \bm{v}) \triangleq f(\bm{x}) + g(\bm{x}) \pi(\bm{x}, \bm{v})$. The function $\pi$ is called the \emph{prestabilizing controller}.

For all $\bm{v} \in \R^{\ell}$, there exist an open neighborhood of $\bar{x}(\bm{v})$, denoted $\sD_{\bm{v}}\subset \R^n$, and comparison functions $\sigma \in \cK\cL$ and $\gamma \in \cK$ such that for any initial state $\bm{x}_0 \in \R^n$, constant reference $\bm{v} \in \R^{\ell}$, and any disturbance signal $d:[0,T] \to \sW$, any solution $x(t)$ of the prestabilized system subject to disturbances \eqref{eq:sys-disturbance-prestab} satisfies $\forall t \in [0,T]$,
\begin{align*}
  \|x(t) \!-\! \bar{x}(\bm{v})\| &
  \leq \sigma \left(\|\bm{x}_0 \!-\! \bar{x}(\bm{v})\|, t\right) +
  \gamma \left(d_{\max}\right)
\end{align*}
where $d_{\max} \triangleq \max_{\bm{d} \in \sW} \|\bm{d}\|$.

Recalling the set of state constraints $\sX \subset \R^n$ that the system \eqref{eq:sys-disturbance} must respect, let us define some related sets. The steady-state admissible reference set
\begin{equation}
  \sV = \{\bm{v} \in \R^{\ell} : \bar{x}(\bm{v}) \in \sX, \bar{u}(\bm{v}) \in \sU\},
\end{equation}
is closed in $\R^{\ell}$. For each reference $\bm{v} \in \R^{\ell}$, the input constraint set $\sU$ induces state constraints on the prestabilized system $f_{\pi}$. This allows us to define the prestabilized constraint set
\begin{equation} \label{eq:set-prestab-set}
  \sX_{\bm{v}} = \{\bm{x} \in \sX \mid \pi(\bm{x}, \bm{v}) \in \sU\},
\end{equation}
which is closed in $\R^n$.
In upcoming results, we will consider the augmented state-reference space obtained by concatenating $(\bm{x}, \bm{v}) \in \R^{n} \times \R^{\ell}$. We will generally denote sets in this space with a tilde (e.g. $\tilde{\sX}$) accent to distinguish them from their state or reference space counterparts. The augmented constraint set
\begin{equation}
  \tilde{\sX} = \{(\bm{x}, \bm{v}) \in \R^n \times \R^{\ell} :\bm{x} \in \sX_{\bm{v}}, \bm{v} \in \sV\},
\end{equation}
is closed in $\R^n \times \R^{\ell}$.

\begin{assumption} \label{ass:set-augmented-roa-open}
    The augmented region of attraction
    \begin{equation}
      \tilde{\sD} = \{(\bm{x}, \bm{v}) \in \R^n \times \R^{\ell} \mid \bm{x} \in    \sD_{\bm{v}}\},
    \end{equation}
    is open in $\R^n \times \R^{\ell}$.  
\end{assumption}
Given a function $\Delta_i:\R^n \times \R^{\ell} \to \R^q$, let us define its augmented active constraint set at $(\bm{x}, \bm{v})$ by
\begin{equation}
  \tilde{\sI}(\bm{x}, \bm{v}) = \left\{i \in \{1,\ldots, q\} \mid \Delta_i(\bm{x}, \bm{v}) = 0 \right\}.
\end{equation}

With these pieces, we can define robust dynamic safety margins, which measure the distance to constraint violation (under worst-case disturbances) of the prestabilized system subject to disturbances \eqref{eq:sys-disturbance-prestab} for a constant reference $\bm{v}$.
\begin{definition} \label{def:rdsm}
  Let $\Delta: \tilde{\sD} \to \R^{q}$ be continuously differentiable with superlevel set $\tilde{\sC} = \{(\bm{x}, \bm{v}) \in \tilde{\sD} : \Delta(\bm{x}, \bm{v}) \geq 0\}$. The function $\Delta$ is called a \emph{robust dynamic safety margin} (RDSM) if the following conditions hold:
  \begin{subequations}
    \begin{align}
      & \tilde{\sC} \subset \tilde{\sX}, \label{eq:rdsm-safe}\\
      & \tilde{\sC} ~ \text{is compact in} ~ \R^n \times \R^{\ell}, \label{eq:rdsm-compact}\\
      & \forall (\bm{x}, \bm{v}) \in \partial \tilde{\sC}, \min_{(i,d) \in \tilde{\sI}(\bm{x}, \bm{v}) \times \sW} \!\!\!\!\! D_1 \Delta_i(\bm{x}, \bm{v}) \dot{\bm{x}} \geq 0, \label{eq:rdsm-invariant}
    \end{align}
  \end{subequations}
  where $\dot{\bm{x}}$ is given in \eqref{eq:sys-disturbance-prestab}.
\end{definition}
Definition \ref{def:rdsm} coincides with the definition of nominal dynamic safety margins when $\sW = \{0\}$ (i.e., no disturbances).

\subsection{RDSMs are RCBFs}
While Definition \ref{def:rdsm} appears opaque at first glance, the RG literature has developed tools to construct RDSMs  \cite{nicotra2016explicit}. Before we review these tools, we show that RDSMs are RCBFs for an augmented system consisting of the concatenation of $\bm{x} \in \R^n$ and $\bm{v} \in \R^{\ell}$ given by
\begin{equation}\label{eq:sys-disturbance-aug}
  \mat{\dot{\bm{x}}\\ \dot{\bm{v}}} = \mat{f(\bm{x}) + g(\bm{x}) \bm{u} +
  p(\bm{x}) \bm{d} \\ \bm{w}},
\end{equation}
where $(\bm{u}, \bm{w}) \in \R^m \times \R^{\ell}$ are the augmented inputs.

\begin{theorem}[\bf{RDSMs are RCBFs}] \label{thm:rdsm-rcbf}
  If $\Delta: \tilde{\sD} \to \R^{q}$ is an RDSM, then $\Delta$ is an RCBF with respect to the augmented system \eqref{eq:sys-disturbance-aug}.
\end{theorem}
\begin{proof}
  Pick any $i \in \{1, \ldots, q\}$ and define the set-valued map $\hat{\sC}_i: [0, \infty) \rightrightarrows \tilde{\sD}$ by
  \begin{equation*}
    \hat{\sC}_i(c) = \{(\bm{x}, \bm{v}) \in \tilde{\sC} \mid 0 \leq
    \Delta_i(\bm{x}, \bm{v}) \leq c\},
  \end{equation*}
  the function $\tilde{\xi}_i: \tilde{\sD} \times \R^{n_d} \to \R$ by
  \begin{equation*}
    \tilde{\xi}_i(\bm{x}, \bm{v}, \bm{d}) = D_1 \Delta_i(\bm{x}, \bm{v})
    \big(f_{\pi}(\bm{x}, \bm{v}) + p(\bm{x}) \bm{d}\big),
  \end{equation*}
  and the value function $\hat{\alpha}_i: [0,\infty) \to \R$ by
  \begin{equation*}
    \hat{\alpha}_i(c) = - \inf_{(\bm{x}, \bm{v}, \bm{d}) \in \hat{\sC}_i(c)
    \times \sW} \tilde{\xi}_i(\bm{x}, \bm{v}, \bm{d}).
  \end{equation*}
  We seek to upper-bound $\hat{\alpha}_i$ by a class $\cK$ function as in the proof of Theorem \ref{thm:rcbf-rci}.

  For the first condition ($\hat{\alpha}_i$ upper semicontinuous), it can be shown that $c \mapsto \hat{\sC}_i \times \sW$ is upper semicontinuous. Recalling that $c \mapsto \hat{\sC}_i(c) \times \sW$ has compact values, it follows by the Berge maximum theorem \cite{berge1963topological} that the value function $\hat{\alpha}_i$ is upper semicontinuous.

  For the second condition ($\hat{\alpha}(0) \leq 0$), pick any $(\bm{x}, \bm{v}, \bm{d}) \in \hat{\sC}_i(0) \times \sW \subset \partial \tilde{\sC} \times \sW$ and note that since $i \in \tilde{\sI}(\bm{x}, \bm{v})$, it follows by \eqref{eq:rdsm-invariant} that $\tilde{\xi}_i(\bm{x}, \bm{v}, \bm{d}) \geq 0$, and thus $\hat{\alpha}_i(0) \leq 0$.

  Having shown its conditions hold, it follows by \cite[Lemma 4.1]{albertini1999continuous} that there exists $\alpha_i \in \cK$ that upper-bounds $\hat{\alpha}_i$.  Since $i \in \{1,\ldots, q\}$ was arbitrary, let us define $\alpha(c) = \max_{i \in \{1,\ldots,q\}} \alpha_i(c)$ and note that $\alpha \in \cK$. 
  
  Finally, let $(\bm{x}, \bm{v}) \in \tilde{\sC}$ be given and recall that by \eqref{eq:rdsm-safe}, $\pi(\bm{x},\bm{v}) \in \sU$. For any $i \in \{1,\ldots,q\}$ and any $\bm{d} \in \sW$,
  \begin{align*}
    \tilde{\xi}_i(\bm{x}, \bm{v}, \bm{d}) &= D\Delta_i(\bm{x}, \bm{v}) \mat{f(\bm{x}) + g(\bm{x}) \pi(\bm{x}, \bm{v}) +  p(\bm{x}) \bm{d} \\ 0} \\ 
    &\geq \inf_{(\bm{x}_0, \bm{v}_0, \bm{d}_0) \in \hat{\sC}_i(\Delta_i(\bm{x}, \bm{v})) \times \sW} \tilde{\xi}_i(\bm{x}_0, \bm{v}_0, \bm{d}_0) \\
    &\geq -\hat{\alpha}_i\big(\Delta_i(\bm{x}, \bm{v})\big)\\
    &\geq -\alpha_i\big(\Delta_i(\bm{x}, \bm{v})\big)\\
    &\geq -\alpha\big(\Delta_i(\bm{x}, \bm{v})\big).
  \end{align*}
  We have shown that $\alpha \in \cK$ is such that for all $(\bm{x}, \bm{v}) \in \tilde{\sC}$, there exists $(\bm{u}, \bm{w}) \in \sU \times \R^{\ell}$ such that $\forall \bm{d} \in \sW$,
  \begin{equation} \label{eq:rdsm-rcbf-cond}
    \min_{i \in \{1, \ldots, q\}}    \left[\dot{\Delta}_i (\bm{x}, \bm{v}, \bm{u}, \bm{w}, \bm{d}) + \alpha\big(\Delta_i(\bm{x}, \bm{v})\big)\right] \geq 0,
  \end{equation}
  where
  \begin{equation}
    \dot{\Delta}_i = D\Delta_i(\bm{x},
    \bm{v}) \mat{f(\bm{x}) + g(\bm{x}) \bm{u} + p(\bm{x}) \bm{d} \\ \bm{w}}.
  \end{equation}
  We conclude that $\Delta$ is an RCBF.
\end{proof}

\begin{definition}
  When we intend to use an RDSM $\Delta$ as an RCBF, as opposed to a component of the robust ERG \cite{nicotra2016explicit}, we refer to it as an \emph{RDSM-RCBF}.
\end{definition}

The following result follows by Theorems \ref{thm:rcbf-rci} and
\ref{thm:rdsm-rcbf}.
\begin{corollary} \label{ct-cor:rdsm-rci}
  Let $\Delta: \tilde{\sD} \to \R^{q}$ be an RDSM-RCBF. If the MFCQ holds, then the superlevel set $\tilde{\sC} \subset \tilde{\sX} \subset \sX \times \sV$ is RCI with respect to the augmented system \eqref{eq:sys-disturbance-aug}.
\end{corollary}

RDSMs being RCBFs guarantees the existence of $\alpha \in \cK$ such that the RDSM-RCBF condition \eqref{eq:rdsm-rcbf-cond} holds. While this is useful, the condition can be strengthened given our knowledge that the prestabilizing controller $\pi$ is exactly one of the safe inputs that satisfies the RDSM-RCBF condition.
\begin{corollary} \label{cor:rdsm-rcbf-cond-pi}
  Let $\Delta: \tilde{\sD} \to \R^{q}$ be an RDSM-RCBF with superlevel set $\tilde{\sC}$. There exists $\alpha \in \cK$ for which $\forall (\bm{x}, \bm{v}) \in \tilde{\sC}$, $\forall \bm{d} \in \sW$,
  \begin{equation} \label{eq:rdsm-rcbf-cond-pi}
    \min_{i \in \{1,\ldots,q\}} \left[ \dot{\Delta}_i\big(\bm{x}, \bm{v}, \pi(\bm{x}, \bm{v}), 0, \bm{d}\big) + \alpha\big(\Delta_i(\bm{x}, \bm{v})\big)\right] \geq 0.
  \end{equation}
  Moreover, if the MFCQ holds, then
  \begin{equation}
    \forall (\bm{x}, \bm{v}) \in \tilde{\sC}, \quad \big(\pi(\bm{x}, \bm{v}), 0\big) \in \sK_{\tilde{\sC}}(\bm{x}, \bm{v}).
  \end{equation}
  In other words, $\pi$ is a continuous selection of the robust regulation map $\sK_{\tilde{\sC}}$ given in \eqref{eq:rregulation-map}.
\end{corollary}
\begin{proof}
  The first statement was already shown in the proof of Theorem  \ref{thm:rdsm-rcbf}. The second statement follows by Proposition  \ref{prop:contingent-cone-practical} because $\tilde{\sC}$ is practical when the MFCQ holds.
\end{proof}

\subsection{RDSM-based Safety Filter}
Since RDSMs are RCBFs, we can use them directly in the RCBF-based safety filter \eqref{eq:rcbf-sf}. In the augmented state-reference space, the safe augmented input $(\bm{u}, \bm{w})$ is obtained by solving the following optimization problem

\begin{align} \label{eq:rdsm-rcbf-sf}
  \min_{(\bm{u}, \bm{w}) \in \sU \times \R^{\ell}} \quad & \|\bm{u} - \kappa(\bm{x})\|^2 + \eta \|\bm{w} - \rho(\bm{v})\|^2 \\
  \text{s.t.} \quad & \dot{\Delta}_i(\bm{x}, \bm{v}, \bm{u}, \bm{w}, 0)  \!+\! \tilde{\omega}_i(\bm{x}, \bm{v}) \!\geq\! -\alpha\big(\Delta_i(\bm{x}, \bm{v})\big), \nonumber \\ &  i = 1,\ldots,q, \nonumber
\end{align}
where $\eta > 0$ is a small positive number that makes the objective strictly convex,
\begin{equation}
  \tilde{\omega}_i(\bm{x}, \bm{v}) = \min_{\bm{d} \in \sW} D_1
  \Delta_i(\bm{x}, \bm{v}) p(\bm{x}) \bm{d},
\end{equation}
and $\rho:\R^{\ell} \to \R^{\ell}$ is called the \emph{navigation field} and dictates the desired performance of the reference state $\bm{v}$. Usually, given a desired reference $\bm{r}$, $\rho(\bm{v}) = \bm{r} - \bm{v}$ is a valid navigation field. Another interesting choice when available is $\rho(\bm{v}) = \bar{x}^{-1}(\bm{x}) - \bm{v}$, which points the reference $\bm{v}$ to the reference associated with the equilibrium point ``closest'' to the current state $\bm{x}$. 

In the following sections, we study two approaches to construct robust dynamic safety margins: One leveraging ISS Lyapunov functions, and the other leveraging nominal trajectory predictions and their contractive properties.

\section{Lyapunov Robust Dynamic Safety Margin} \label{sec:lyap-rdsm}
This subsection shows a specific method to construct RDSMs that relies on ISS Lyapunov functions, based on the results from \cite{nicotra2016explicit}. The following ISS Lyapunov function definition is based on \cite[Definition 2.2]{sontag1995characterizations}, but considers general equilibrium points $\bar{x}(\bm{v}) \neq 0$ and allows negative semidefinite derivative $\dot{V} \leq 0$ instead of $\dot{V} < 0$.
\begin{definition} \label{def:lyap-iss}
  A continuously differentiable function $V:\R^n \times \R^{\ell} \to \R$ is an \emph{ISS reference-dependent Lyapunov function} for the prestabilized system subject to disturbances \eqref{eq:sys-disturbance-prestab} if $\forall \bm{v} \in \R^{\ell}$, there exist functions $\gamma_1, \gamma_2,  \chi \in \cK$, and a neighborhood of $\bar{x}(\bm{v})$, denoted $\sD_{\bm{v}} \subset \R^n$, wherein
  \begin{subequations}
    \begin{align}
      &\gamma_1\big(\|\bm{x} - \bar{x}(\bm{v})\|\big) \leq V(\bm{x}, \bm{v})
      \leq \gamma_2\big(\|\bm{x} - \bar{x}(\bm{v})\|\big),
      \label{eq:lyap-iss-bounded} \\
      &\forall \bm{d} \in \R^{n_d} ~ \text{satisfying} ~ \chi(\|\bm{d}\|) \leq
      V(\bm{x}, \bm{v}), \nonumber \\
      &\qquad \quad D_1 V(\bm{x},
      \bm{v}) \left(f_{\pi}(\bm{x},\bm{v}) + p(\bm{x}) \bm{d}
      \right) \leq 0.  \label{eq:lyap-iss-decreases}
    \end{align}
  \end{subequations}
\end{definition}
Additionally, we require $V$ has non-vanishing gradient near the equilibrium points with the following assumption.
\begin{assumption}\label{ass:lyap-regular}
  Given any $\bm{v} \in \R^{\ell}$, $\bm{x} \in \sD_{\bm{v}}\setminus \{\bar{x}(\bm{v})\}$ is a regular point of $V(\cdot, \bm{v}): \R^n \to \R$.
\end{assumption}

Let us recall from \cite{freire2026using} the \emph{safety threshold value} $\Gamma^*:\R^{\ell} \to \R$ defined as
\begin{equation}
  \Gamma^*(\bm{v}) = \inf_{\bm{x} \in \sX_{\bm{v}}^c \cap \sD_{\bm{v}}} V(\bm{x}, \bm{v}) - \inf_{\bm{x} \in \sX_{\bm{v}} \cap \sD_{\bm{v}}} V(\bm{x}, \bm{v}),
\end{equation}
where $\sX_{\bm{v}}^c \triangleq \R^n \setminus \sX_{\bm{v}}$, and the \emph{stability threshold value} $\overline{\Gamma}: \R^{\ell} \to \R$ defined as
\begin{equation}
  \overline{\Gamma}(\bm{v}) = \inf_{\bm{x} \in \partial \sD_{\bm{v}}} V(\bm{x}, \bm{v}).
\end{equation}
With them, we can construct RDSMs using ISS reference-dependent Lyapunov functions.

\begin{theorem}[Lyapunov RDSM] 
  \label{thm:lyap-rdsm}
  Let $V:\R^n \times \R^{\ell} \to \R$ be an ISS reference-dependent Lyapunov function for the prestabilized system subject to disturbances \eqref{eq:sys-disturbance-prestab} satisfying Assumption \ref{ass:lyap-regular}. If $\Gamma^*: \R^{\ell} \to \R$ and
  $\overline{\Gamma}: \R^{\ell} \to \R$ are continuously differentiable, and
  $\sX \subset \R^n$ and $\sV \subset \R^{\ell}$ are compact, then for any
  $\epsilon \in (0,1)$,
  \begin{equation} \label{eq:lyap-rdsm}
    \Delta(\bm{x}, \bm{v}) = \mat{\Gamma^*(\bm{v}) - V(\bm{x}, \bm{v}) \\
    (1-\epsilon)\overline{\Gamma}(\bm{v}) - V(\bm{x}, \bm{v}) \\
    \Gamma^*(\bm{v}) - \chi(d_{\max}) \\
    \overline{\Gamma}(\bm{v}) - \chi(d_{\max})},
  \end{equation}
  is a robust dynamic safety margin.
\end{theorem}
\begin{proof}
  To prove property \eqref{eq:rdsm-safe}, let $(\bm{x}, \bm{v}) \in \tilde{\sC} \subset \tilde{\sD}$. Then, $\Gamma^*(\bm{v}) \geq V(\bm{x}, \bm{v}) \geq 0$, which implies that $\bm{v} \in \sV$ by \cite[Lemma 2]{freire2026using}. Next, for a contradiction, assume $\bm{x} \notin \sX_{\bm{v}}$. Then
  \begin{equation*}
    \Gamma^*(\bm{v}) = \inf_{\bm{x} \in \sX_{\bm{v}}^{c}\cap \sD_{\bm{v}}} V(\bm{x}, \bm{v}) \leq V(\bm{x}, \bm{v}).
  \end{equation*}
  So, we must have $V(\bm{x}, \bm{v}) = \Gamma^*(\bm{v})$. If $V(\bm{x},\bm{v}) = 0$, then $\bm{x} = \bar{x}(\bm{v}) \in \sX_{\bm{v}}$ because $\bm{v} \in \sV$, and we have a contradiction.  Otherwise, if $\Gamma^*(\bm{v}) = V(\bm{x},\bm{v}) > 0$, recall that $\sX^c_{\bm{v}} \cap \sD_{\bm{v}}$ is open in $\R^n$. Therefore, there exists an open set $\sN \subset \sX^c_{\bm{v}} \cap \sD_{\bm{v}}$ with $\bm{x} \in \sN$.
  Furthermore, since $\bm{x} \in \sD_{\bm{v}} \setminus \{\bar{x}(\bm{v})\}$ is a regular point of $V(\cdot, \bm{v})$ by Assumption \ref{ass:lyap-regular}, we have that $D_1 V(\bm{x}, \bm{v}) \neq 0$. 
  So, there exists a sufficiently small scalar $\delta > 0$ such that $\hat{\bm{x}} = \bm{x} - \delta D_1 V(\bm{x}, \bm{v})^\top  \in \sN \subset \sX^c_{\bm{v}} \cap \sD_{\bm{v}}$.  
  Furthermore, this point achieves lower value in $V$. That is, $V(\hat{\bm{x}}, \bm{v}) < V(\bm{x}, \bm{v})$. So, $\Gamma^*(\bm{v}) = V(\bm{x},\bm{v})$ is a contradiction because
  \begin{equation*}
    \Gamma^*(\bm{v}) = \inf_{\bm{x} \in \sX^c_{\bm{v}} \cap \sD_{\bm{v}}} V(\bm{x},\bm{v}) \leq V(\hat{\bm{x}}, \bm{v}) < V(\bm{x}, \bm{v}).
  \end{equation*}

  To prove property \eqref{eq:rdsm-compact}, note that since $\tilde{\sC} \subset \sX \times \sV$, and $\sX \times \sV$ is compact in $\R^n \times \R^{\ell}$, we must only show that $\tilde{\sC}$ is closed in $\R^n \times \R^{\ell}$.
  Let $\{(\bm{x}_k, \bm{v}_k)\} \subset \tilde{\sC}$ be a sequence that converges to $(\bm{x}, \bm{v}) \in \R^n \times \R^{\ell}$. Since $\sX \subset \R^n$ and $\sV \subset \R^{\ell}$ are closed, it follows that $\bm{x} \in \sX$ and $\bm{v} \in \sV$.
  Now we argue $\Delta(\bm{x}, \bm{v}) \geq 0$. For a contradiction, assume there is $i \in \{1,2,3,4\}$ such that $\Delta_i(\bm{x}, \bm{v}) < 0$. By continuity of $\Delta_i$ on $\R^n \times \R^{\ell}$, there exists an open neighborhood $\tilde{\sN} \subset \R^n \times \R^{\ell}$ of $(\bm{x}, \bm{v})$ satisfying $\forall (\hat{\bm{x}}, \hat{\bm{v}}) \in \tilde{\sN}$, $\Delta_i(\hat{\bm{x}}, \hat{\bm{v}}) < 0$. By convergence of $\{(\bm{x}_k, \bm{v}_k)\}$ to $(\bm{x}, \bm{v})$, there exists $K \in \N$ such that $(\bm{x}_K, \bm{v}_K) \in \tilde{\sN}$, which raises a contradiction because $(\bm{x}_K, \bm{v}_K) \in \tilde{\sC}$ implies that $\Delta(\bm{x}_K, \bm{v}_K) \geq 0$.
  It remains to show that $(\bm{x}, \bm{v}) \in \tilde{\sD}$. Begin by noting that for all $k \in \N$, $(\bm{x}_k, \bm{v}_k) \in \tilde{\sD}$ implies that $V(\bm{x}_k, \bm{v}_k) \geq \gamma_1\big(\|\bm{x}_k - \bar{x}(\bm{v}_k)\|\big)$. It follows by continuity of all functions involved in the inequality and by the order limit theorem that $V(\bm{x}, \bm{v}) \geq \gamma_1\big(\|\bm{x} - \bar{x}(\bm{v})\|\big)$. For a contradiction, let us assume that $(\bm{x}, \bm{v}) \in \partial \tilde{\sD}$, which implies that $\bm{x} \in \partial \sD_{\bm{v}}$ and, therefore, $\overline{\Gamma}(\bm{v}) \leq V(\bm{x}, \bm{v})$. However, recall that we have shown $\Delta(\bm{x}, \bm{v}) \geq 0$. Thus, $V(\bm{x}, \bm{v})  \leq (1-\epsilon) \overline{\Gamma}(\bm{v})$. This yields
  \begin{equation*}
    \frac{1}{1-\epsilon} V(\bm{x}, \bm{v}) \leq \overline{\Gamma}(\bm{v}) \leq V(\bm{x}, \bm{v}),
  \end{equation*}
  which is a contradiction because $V(\bm{x}, \bm{v}) \geq \gamma_1\big(\|\bm{x} - \bar{x}(\bm{v})\|\big) > 0$. So, we must have that $(\bm{x}, \bm{v}) \in \tilde{\sD}$ and we conclude that $\tilde{\sC}$ is  closed in $\R^n \times \R^{\ell}$. Since it is also bounded, $\tilde{\sC}$ is compact.

  To prove property \eqref{eq:rdsm-invariant}, let $(\bm{x}, \bm{v}) \in \partial \tilde{\sC} \subset \tilde{\sD}$ be given and pick any $i \in \tilde{\sI}(\bm{x}, \bm{v})$ and any $\bm{d} \in \sW$. 
  If $i = 1$, we have that $\Gamma^*(\bm{v}) = V(\bm{x},\bm{v})$ and $\Delta_3(\bm{x}, \bm{v}) \geq 0$ implies that $V(\bm{x}, \bm{v}) \geq \chi(d_{\max}) \geq \chi(\|\bm{d}\|)$.
  If $i = 2$, we have that $(1-\epsilon) \overline{\Gamma}(\bm{v}) = V(\bm{x}, \bm{v})$ and $\Delta_4(\bm{x}, \bm{v}) \geq 0$ implies that $\overline{\Gamma}(\bm{v}) \geq (1 -\epsilon) \overline{\Gamma}(\bm{v}) = V(\bm{x}, \bm{v}) \geq \chi(d_{\max}) \geq \chi(\|\bm{d}\|)$.
  In either case ($i \in \{1,2\}$), by \eqref{eq:lyap-iss-decreases},
  \begin{align*}
    &D_1 \Delta_i(\bm{x}, \bm{v}) \left(f_{\pi}(\bm{x}, \bm{v}) + p(\bm{x}) \bm{d} \right) \\
    & \qquad \qquad = - D_1 V(\bm{x}, \bm{v}) \big(f_{\pi}(\bm{x}, \bm{v}) +  p(\bm{x}) \bm{d}\big) \geq 0.
  \end{align*}
  On the other hand, if $i \in \{3,4\}$,
  \begin{equation*}
    D_1 \Delta_i(\bm{x}, \bm{v}) \left(f_{\pi}(\bm{x}, \bm{v}) + p(\bm{x}) \bm{d} \right) = 0.
  \end{equation*}
\end{proof}
\begin{remark}
  If Assumption \ref{ass:roa-large} holds, then the stability threshold value $\overline{\Gamma}$ plays no role and we can replace \eqref{eq:lyap-rdsm} with the simpler expression
  \begin{equation}
    \Delta(\bm{x}, \bm{v}) = \mat{\Gamma^*(\bm{v}) - V(\bm{x}, \bm{v}) \\
    \Gamma^*(\bm{v}) - \chi(d_{\max})}.
  \end{equation}
\end{remark}
Theorem \ref{thm:lyap-rdsm} links the construction of an RDSM, and therefore, by Theorem \ref{thm:rdsm-rcbf}, the construction of an RCBF, to the problem of finding an ISS Lyapunov function. The following result shows that the associated RCI set $\tilde{\sC}$ contains robustly steady-state admissible equilibria.
\begin{lemma} \label{lem:lyap-rdsm-nonempty}
  Let $\Delta: \tilde{\sD} \to \R^4$ be a Lyapunov RDSM as in
  \eqref{eq:lyap-rdsm} and define the set
  \begin{equation*}
    \sV_d = \left\{\bm{v} \in \R^{\ell} \mid \min\big(\Gamma^*(\bm{v}), \overline{\Gamma}(\bm{v}) \big) \geq \chi(d_{\max}) \right\} \subset \sV.
  \end{equation*}
  If $\bm{v} \in \sV_d$, then $\big(\bar{x}(\bm{v}), \bm{v}\big) \in
  \tilde{\sC}$.
\end{lemma}
\begin{proof}
  Let $\bm{v} \in \sV_d$ be given. Clearly, $\Gamma^*(\bm{v}) \geq \chi(d_{\max}) \geq 0$, which implies by \cite[Lemma 2]{freire2026using} that $\bm{v} \in \sV$.
  Also, $\bar{x}(\bm{v}) \in \sD_{\bm{v}}$, which implies that $(\bar{x}(\bm{v}), \bm{v}) \in \tilde{\sD}$. 
  Finally,
  \begin{align*}
    \Delta_1\big(\bar{x}(\bm{v}), \bm{v}\big) &= \Gamma^*(\bm{v}) \geq \chi(d_{\max}) \geq 0\\
    \Delta_2\big(\bar{x}(\bm{v}), \bm{v}\big) &= (1 - \epsilon) \overline{\Gamma}(\bm{v}) \geq (1 - \epsilon) \chi(d_{\max})    \geq 0 \\
    \Delta_3\big(\bar{x}(\bm{v}), \bm{v}\big) &= \Gamma^*(\bm{v}) - \chi(d_{\max}) \geq 0 \\
    \Delta_4\big(\bar{x}(\bm{v}), \bm{v}\big) &= \overline{\Gamma}(\bm{v}) - \chi(d_{\max}) \geq 0,
  \end{align*}
  where we have used $V\big(\bar{x}(\bm{v}), \bm{v}\big) = 0$.
\end{proof}

The following corollary highlights that for Lyapunov RDSMs, any $\alpha \in \cK$ can be used in the RDSM-RCBF condition \eqref{eq:rdsm-rcbf-cond}.
\begin{corollary} \label{cor:lyap-rdsm-any-alpha-works}
  If $\Delta$ is a Lyapunov RDSM as in \eqref{eq:lyap-rdsm}, then the RDSM-RCBF condition \eqref{eq:rdsm-rcbf-cond} holds for any $\alpha \in
  \cK$.
\end{corollary}
\begin{proof}
  Let $\alpha \in \cK$ and $(\bm{x}, \bm{v}) \in \tilde{\sC}$ be given. Recall that $\pi(\bm{x}, \bm{v}) \in \sU$ and note that for any $i \in \{3,4\}$ and any $\bm{d} \in \sW$, $\dot{\Delta}_i\big(\bm{x}, \bm{v}, \pi(\bm{x}, \bm{v}), 0, \bm{d}\big) = 0 \geq -\alpha\big(\Delta_i(\bm{x}, \bm{v})\big)$.
  Moreover, for any $\bm{d} \in \sW$, $\Delta_1(\bm{x}, \bm{v}) \geq 0$ and $\Delta_3(\bm{x}, \bm{v}) \geq 0$ imply that $V(\bm{x}, \bm{v}) \geq \chi(d_{\max}) \geq \chi(\|\bm{d}\|)$. Therefore, it follows by \eqref{eq:lyap-iss-decreases} that for any $i \in \{1,2\}$,
  \begin{align*}
    \dot{\Delta}_i\big(\bm{x}, \bm{v}, \pi(\bm{x}, \bm{v}), 0, \bm{d}\big) &=
    - D_1 V(\bm{x}, \bm{v}) \big(f_{\pi}(\bm{x}, \bm{v}) + p(\bm{x})
    \bm{d}\big) \\
    &\geq 0 \\
    &\geq -
    \alpha\big(\Delta_i(\bm{x}, \bm{v})\big).
  \end{align*}
\end{proof}

\section{Trajectory Robust Dynamic Safety Margin} \label{sec:traj-rdsm}

In this section, we evaluate another method to construct RDSMs that relies on trajectory predictions. To apply this method, we must make the following additional assumptions.
\begin{assumption} \label{ass:sys-prestab-differentiable}
  The map $f_{\pi}$ is continuously differentiable.
\end{assumption}
\begin{assumption} 
  \label{ass:cst-differentiable}
  There exists a continuously differentiable function $s:\R^n \times \R^{\ell} \to \R^q$ describing the reference-de\-pen\-dent constraint set. That is,
  \begin{equation}
    \forall \bm{v} \in \R^{\ell}, \; \sX_{\bm{v}} = \Big\{\bm{x} \in \R^n \;\big|\;  \min_{i \in \{1,\ldots,q\}} s_i(\bm{x}, \bm{v}) \geq 0\Big\}.
  \end{equation}
\end{assumption}
To simplify the presentation, we also assume the following.
\begin{assumption} 
  \label{ass:roa-large}
  The constraints satisfy $\tilde{\sX} \subset \tilde{\sD}$.
\end{assumption}
Indeed, we make Assumption~\ref{ass:roa-large} without loss of generality because it can always be enforced by introducing state constraints that capture the stability requirements of the prestabilizing controller. The additional Assumptions~\ref{ass:sys-prestab-differentiable}--\ref{ass:roa-large} were also made in the presentation of the nominal trajectory DSM \cite{freire2026designing}, and ensure the state $S_x: t \mapsto D_2\phi(\cdot,\bm{x},\bm{v})$ and reference $S_v: t \mapsto D_3\phi(\cdot, \bm{x}, \bm{v})$ sensitivity matrices are well defined. To make the result robust, we strengthen Assumption \ref{ass:iss} by requiring contractivity \cite{bullo2026contraction} beyond ISS.
\begin{assumption}[Contractivity] \label{ass:contractivity}
  There exists $P \succ 0$ and $c > 0$ such that, for all $\bm{x}_1, \bm{x}_2 \in \R^n$, $\bm{v} \in \R^{\ell}$, and $\bm{d} \in \R^q$,
  \begin{equation}
    (\dot{\bm{x}}_1 - \dot{\bm{x}}_2)^\top P(\bm{x}_1 - \bm{x}_2) \leq -c \|\bm{x}_1 - \bm{x}_2\|^2_{P},
  \end{equation}
  where $\dot{\bm{x}}_i = f_{\pi}(\bm{x}_i, \bm{v}) + p(\bm{x}_i)\bm{d}$.
\end{assumption}
\begin{assumption}
  \label{ass:bounded-disturbance-function}
  The dis\-turbance function $p(\bm{x})$ is bounded relative to the metric defined by $P \succ 0$ from Assumption \ref{ass:contractivity}. That is, there exists a constant $c_d > 0$ such that for all $\bm{x} \in \R^n$, $\bm{d}_1, \bm{d}_2 \in \R^{n_d}$,
  \begin{equation}
    \|p(\bm{x})(\bm{d}_1 - \bm{d}_2)\|_{P} \leq c_d \|\bm{d}_1 - \bm{d}_2\|.
  \end{equation}
\end{assumption}
Note that Assumption~\ref{ass:bounded-disturbance-function} always holds for constant functions $p(\bm{x}) = E$ (as in linear systems). Let $\phi(\tau, \bm{x}, \bm{v})$ denote the unique (by Assumption~\ref{ass:sys-prestab-differentiable} and \cite[Theorem~3.1]{khalil2002nonlinear}) solution to the initial value problem
\begin{equation} \label{eq:ivp-prestab}
  D_1 \phi(\tau, \bm{x}, \bm{v}) = f_{\pi}\big(\phi(\tau, \bm{x}, \bm{v}), \bm{v}), \quad \phi(0, \bm{x}, \bm{v}) = \bm{x}.
\end{equation}
When there is no ambiguity, we write $\phi(\tau) = \phi(\tau, \bm{x}) = \phi(\tau, \bm{x}, \bm{v})$.
Note that if $\bm{x} \in \sD_{\bm{v}}$, it follows by \cite[Theorem~3.3]{khalil2002nonlinear} that the solution $\phi(\cdot, \bm{x}, \bm{v})$ is globally defined on $[0,\infty)$.

The following theorem presents the main contractivity result that we will leverage in this section.
\begin{theorem}[{\cite[Theorem~3.16]{bullo2026contraction}}] \label{thm:contractivity}
  Under Assumptions~\ref{ass:compact-convex-disturbance-set}, \ref{ass:contractivity} and \ref{ass:bounded-disturbance-function}, and given $\bm{x} \in \R^n$ and $\bm{v} \in \R^{\ell}$, any solution $x(t)$ to \eqref{eq:sys-disturbance-prestab} satisfies
  \begin{equation*}
    \|x(t) - \phi(t)\|_P \leq e^{-ct} \|x(0) - \bm{x}\|_P + \frac{c_d}{c}\left(1 - e^{-ct}\right) d_{\max}.
  \end{equation*}
\end{theorem}

\begin{lemma}\label{lem:contraction-tube-continuous}
  The contraction tube $\sR: [0,\infty) \rightrightarrows \R^n$ defined
  \begin{equation} \label{eq:contraction-tube}
    \sR(\tau) = \left\{\bm{x} \in \R^n \;\Big|\; \|\bm{x}\|_P \leq \frac{c_d}{c}(1 \!-\! e^{-c \tau})d_{\max}\right\},
  \end{equation}
  is continuous (upper and lower semicontinuous) on $[0,\infty)$.
\end{lemma}
\begin{proof}
  Let $\tau \in [0,\infty)$ be given and pick any $\bm{x} \in \sR(\tau)$ and any sequence $\{\tau_k\} \subset [0,\infty)$ that converges to $\tau$. Then, define the sequence $\{\bm{x}_k\}$ where $\bm{x}_k = \argmin_{\bm{x}' \in \sR(\tau_k)} \|\bm{x} - \bm{x}'\|_P$. Note the sequence is well defined because $\sR(\tau_k)$ is compact and convex, and $\|\bm{x} - \bm{x}'\|_P$ is strictly convex and continuous. Therefore, $\bm{x}_k$ is unique. Since the constructed sequence satisfies $\bm{x}_k \in \sR(\tau_k)$ and converges to $\bm{x}$, we conclude that $\sR$ is lower semicontinuous on $[0, \infty)$.
  
  Again, let $\tau \in [0,\infty)$ be given and pick any sequence $\{\tau_k\} \subset [0, \infty)$ that converges to $\tau$. Pick also any sequence $\{\bm{x}_k\}$ that converges to $\bm{x} \in \R^n$ and satisfies $\bm{x}_k \in \sR(\tau_k)$. Note that $\forall k \in \N$, $\|\bm{x}_k\|_P \leq \frac{c_d}{c} (1- e^{-c\tau_k}) d_{\max}$. By continuity of both sides of the inequality and convergence of the sequences, we can take the limit as $k \to \infty$ on both sides to obtain that $\|\bm{x}\|_P \leq \frac{c_d}{c}(1-e^{-c\tau})d_{\max}$. Therefore, $\bm{x} \in \sR(\tau)$ and we conclude because $\sR(\tau)$ is compact that $\sR$ is upper semicontinuous on $[0,\infty)$.
\end{proof}

Note that $\sR(0) = \{0\}$ and define also the steady-state contraction tube
\begin{equation} \label{eq:contraction-tube-ss}
  \bar{\sR} = \lim_{\tau \to \infty}\sR(\tau) = \left\{\bm{x} \in \R^n ~\Big|~
  \|\bm{x}\|_P \leq \frac{c_d d_{\max}}{c}\right\}.
\end{equation}

The following theorem presents the trajectory RDSM relying on nominal (no disturbances) trajectory predictions of the prestabilized system \eqref{eq:ivp-prestab} and their contractivity.
\begin{theorem} \label{thm:traj-rdsm}
  Assume $\sX \subset \R^n$ and $\sV \subset \R^{\ell}$ are compact. 
  If the function $\Delta: \tilde{\sD} \to \R^q$, with channels
  \begin{equation} \label{eq:traj-rdsm}
    \Delta_i(\bm{x}, \bm{v}) = \inf_{\tau \in [0,\infty)} ~ \min_{\bm{x}' \in \sR(\tau)} s_i\big(\phi(\tau, \bm{x}, \bm{v}) + \bm{x}', \bm{v} \big),
  \end{equation}
  is continuously differentiable, then it is an RDSM.
\end{theorem}
\begin{proof}
  To prove property \eqref{eq:rdsm-safe}, let $(\bm{x}, \bm{v}) \in \tilde{\sC}$ be given. Then, we have that $\forall i \in \{1,\ldots,q\}$, 
  \begin{align*}
  s_i(\bm{x}, \bm{v}) &= s_i\big(\phi(0), \bm{v}\big)\\
  &= \min_{\bm{x}' \in \sR(0)} s\big(\phi(0) + \bm{x}', \bm{v}\big) \\
  &\geq \Delta(\bm{x}, \bm{v}) \\
  &\geq 0,
  \end{align*}
  which implies $\bm{x} \in \sX_{\bm{v}}$. Next, recall
  that
  \begin{equation*}
    \lim_{\tau \to \infty} \phi(\tau, \bm{x}, \bm{v}) = \bar{x}(\bm{v}).
  \end{equation*}  
  Therefore, there exists $\tau' \in [0,\infty)$ such that $\bar{x}(\bm{v}) \in \phi(\tau', \bm{x}, \bm{v}) + \sR(\tau')$. It follows that $\forall i \in \{1,\ldots q\}$
  \begin{align*}
    s_i\big(\bar{x}(\bm{v}), \bm{v}\big) &\geq \min_{\bm{x}' \in \sR(\tau')}  s( \phi(\tau') +\bm{x}',\bm{v})\\ &\geq \Delta(\bm{x}, \bm{v})\\ &\geq 0,
  \end{align*}
  which implies $\bar{x}(\bm{v}) \in \sX_{\bm{v}}$, and we conclude $\bm{v} \in \sV$.

  To prove property \eqref{eq:rdsm-compact}, note that $\tilde{\sC}$ is bounded since $\tilde{\sC} \subset \tilde{\sX} \subset \sX \times \sV$. Thus, we only need to show that  $\tilde{\sC}$ is closed in $\R^n \times \R^{\ell}$. This property follows by continuity of $\Delta$ and the fact that, by Assumption \ref{ass:roa-large}, $\tilde{\sC} \subset  \tilde{\sX} \subset \tilde{\sD}$, where $\tilde{\sX}$ is closed in $\R^n \times \R^{\ell}$.
  
  To prove property \eqref{eq:rdsm-invariant}, let $(\bm{x}, \bm{v}) \in
  \partial \tilde{\sC}$ be given. For a contradiction, pick $i \in
  \tilde{\sI}(\bm{x}, \bm{v})$ and $\bm{d} \in \sW$ such that
  \begin{equation*}
    D_1 \Delta_i(\bm{x}, \bm{v}) \big(f_{\pi}(\bm{x}, \bm{v}) + p(\bm{x})
    \bm{d} \big) < 0.
  \end{equation*}
  Next, let $\tilde{x}:[0,\infty) \to \R^n$ be the solution to the initial
  value problem
  \begin{equation*}
    D \tilde{x}(\tau) = f_{\pi}\big(\tilde{x}(\tau), \bm{v}\big) +
    p\big(\tilde{x}(\tau)\big)\bm{d}, \quad \tilde{x}(0) = \bm{x}.
  \end{equation*}
  With this, define the function $\xi:t \mapsto \Delta_i\big(\tilde{x}(t),
  \bm{v}\big)$, which is continuously differentiable and satisfies
  \begin{align*}
    \xi(0) &= \Delta_i\big(\tilde{x}(0), \bm{v}\big) = \Delta_i(\bm{x},
    \bm{v}) = 0,\\
    \dot{\xi}(0) &= D_1 \Delta_i(\bm{x}, \bm{v}) \left(
    f_{\pi}(\bm{x}, \bm{v}) + p(\bm{x})\bm{d}\right) < 0.
  \end{align*}
  By continuity of $\xi$, there exists $\epsilon > 0$ such that
  \begin{equation} \label{eq:rdsm-rcbf-contradiction}
    \xi(\epsilon) = \Delta_i\big(\tilde{x}(\epsilon), \bm{v}\big) < 0.
  \end{equation}
  Next, we argue that for any $\sigma \geq 0$,
  \begin{equation*}
    \phi\big(\sigma, \tilde{x}(\epsilon), \bm{v}\big) + \sR(\sigma) \subset
    \phi(\sigma + \epsilon, \bm{x}, \bm{v}) + \sR(\sigma + \epsilon).
  \end{equation*}
  That is, for any $\bm{y} \in \phi\big(\sigma, \tilde{x}(\epsilon),
  \bm{v}\big) + \sR(\sigma)$,
    \begin{align*}
    &  \|\bm{y} - \phi(\sigma + \epsilon, \bm{x})\|_P \\ & \leq \|\bm{y} - \phi\big(\sigma, \tilde{x}(\epsilon)\big)\|_P + \|\phi\big(\sigma, \tilde{x}(\epsilon)\big) - \phi(\sigma + \epsilon, \bm{x})\|_P \\
    &\leq \|\bm{y} - \phi\big(\sigma, \tilde{x}(\epsilon)\big)\|_P + \|\phi\big(\sigma, \tilde{x}(\epsilon)\big) - \phi\big(\sigma, \phi(\epsilon, \bm{x})\big)\|_P \\
    & \leq \|\bm{y} - \phi\big(\sigma, \tilde{x}(\epsilon)\big)\|_P + e^{-c \sigma} \|\tilde{x}(\epsilon) - \phi(\epsilon, \bm{x})\|_P \\
    &\leq \frac{c_d}{c} (1 - e^{-c\sigma})d_{\max} + \frac{c_d}{c} e^{-c\sigma}(1 - e^{-c \epsilon}) d_{\max}\\
    &\leq \frac{c_d}{c} \left(1 - e^{-c(\sigma + \epsilon)}\right)d_{\max},
  \end{align*}
  which implies that $\bm{y} \in \phi(\sigma + \epsilon, \bm{x}, \bm{v}) +
  \sR(\sigma + \epsilon)$. Finally,
  \begin{align*}
    \Delta_i\big(\tilde{x}(\epsilon), \bm{v}\big)\! &= \inf_{\sigma \in [0,\infty)} \min_{\bm{x}' \in \sR(\sigma)} s_i\left(\phi\big(\sigma, \tilde{x}(\epsilon)\big) + \bm{x}', \bm{v}\right) \\
    & \geq \inf_{\sigma \in [0,\infty)} \min_{\bm{x}' \in \sR(\sigma + \epsilon)} s_i\big(\phi(\sigma + \epsilon, \bm{x}) + \bm{x}', \bm{v}\big) \\
    &\geq \inf_{\sigma \in [-\epsilon, \infty)} \min_{\bm{x}'  \in \sR(\sigma + \epsilon)}  s_i\big(\phi(\sigma + \epsilon, \bm{x}) + \bm{x}', \bm{v}\big)\\
    &\geq \inf_{\tau \in [0, \infty)} \min_{\bm{x}'  \in \sR(\tau)}  s_i\big(\phi(\tau, \bm{x}) + \bm{x}', \bm{v}\big)\\
    &\geq \Delta_i(\bm{x}, \bm{v})\\
    &\geq 0
  \end{align*}
  which contradicts \eqref{eq:rdsm-rcbf-contradiction}.
\end{proof}
As in Lemma~\ref{lem:lyap-rdsm-nonempty}, the following characterizes equilibria in $\tilde{\sC}$.
\begin{lemma}
  Let $\Delta: \tilde{\sD} \to \R^q$ be a trajectory RDSM as in \eqref{eq:traj-rdsm} and define the set $\sV_d \subset \sV$ by
  \begin{equation*}
    \sV_d = \left\{\bm{v} \in \R^{\ell} ~\Big|~ \min_{i \in \{1,\ldots, q\}} \; \min_{\bm{x}' \in  \bar{\sR}}  s_i\big(
    \bar{x}(\bm{v}) + \bm{x}', \bm{v}\big) \geq 0 \right\}.
  \end{equation*}
  If $\bm{v} \in \sV_d$, then $\big(\bar{x}(\bm{v}), \bm{v}\big) \in
  \tilde{\sC}$.
\end{lemma}
\begin{proof}
  Let $\bm{v} \in \sV_d$ be given. For any $i \in \{1,\ldots, q\}$,
  \begin{equation*}
    0 \leq \min_{\bm{x}' \in \bar{\sR}} s_i\big(\bar{x}(\bm{v}) + \bm{x}',
    \bm{v}\big)  \leq s_i\big(\bar{x}(\bm{v}), \bm{v}\big),
  \end{equation*}
  which implies $\bar{x}(\bm{v}) \in \sX_{\bm{v}}$, and therefore $\bm{v} \in
  \sV$. Also, $\bar{x}(\bm{v}) \in \sD_{\bm{v}}$, which implies that
  $\big(\bar{x}(\bm{v}), \bm{v}\big) \in \tilde{\sD}$. Finally, let $i \in
  \{1,\ldots, q\}$ be given and note that
  \begin{align*}
    \Delta_i\big(\bar{x}(\bm{v}), \bm{v}\big) &= \inf_{\tau \in [0,\infty)}
    \min_{\bm{x}' \in \sR(\tau)} s_i\Big(\phi\big(\tau, \bar{x}(\bm{v}),
    \bm{v}\big) + \bm{x}', \bm{v}\Big) \\
    &= \inf_{\tau \in [0,\infty)} \min_{\bm{x}' \in \sR(\tau)}
    s_i\big(\bar{x}(\bm{v}) + \bm{x}', \bm{v}\big)\\
    &= \min_{\bm{x}' \in \bar{\sR}} s_i\big(\bar{x}(\bm{v}) + \bm{x}',
    \bm{v}\big)\\
    &\geq 0,
  \end{align*}
  where we have used the fact that $\forall \tau \in [0,\infty)$, $\sR(\tau)
  \subset \bar{\sR}$. Since $i \in \{1,\ldots, q\}$ was arbitrary, we conclude
  that $\Delta\big(\bar{x}(\bm{v}), \bm{v}\big) \geq 0$ and it follows that
  $\big(\bar{x}(\bm{v}), \bm{v}\big) \in \tilde{\sC}$.
\end{proof}

There are two drawbacks that make it difficult to work with \eqref{eq:traj-rdsm}: 1) it can be challenging to solve \eqref{eq:traj-rdsm} on an infinite horizon $\tau \in [0,\infty)$ and 2) continuous differentiability of $\Delta$ is a strong assumption that may not hold in practice. 

\subsection{Horizon Considerations}
To relax the infinite horizon, we consider a finite interval $\tau \in [0,T]$ in \eqref{eq:traj-rdsm} and enforce the predicted trajectory lands in a given robust control invariant terminal set at $\tau = T$.
\begin{proposition} \label{cor:traj-rdsm-fh}
  Let $\Delta_F:\tilde{\sD} \to \R^{q_F}$ be an RDSM, with $\tilde{\sC}_F = \{(\bm{x}, \bm{v}) \in
  \tilde{\sD} \mid \Delta_F(\bm{x}, \bm{v}) \geq 0\}$, satisfying
  \begin{equation} \label{eq:rdsm-invariant-strict}
    \forall (\bm{x}, \bm{v}) \in \partial \tilde{\sC}_F, \; \min_{i \in \tilde{\sI}(\bm{x}, \bm{v})} D_1\Delta_{F,i}(\bm{x}, \bm{v}) f_{\pi}(\bm{x}, \bm{v}) \geq \epsilon,
  \end{equation}
  for some $\epsilon > 0$. Define $\Delta: \tilde{\sD} \to \R^{q + q_F}$ by
  \begin{equation} \label{eq:traj-rdsm-fh}
    \Delta(\bm{x},\! \bm{v}) \!=\!\! \mat{\min_{\tau \in [0,\!T]}  \min_{\bm{x}' \!\in\! \sR(\tau)}
    s_1\big(\phi(\tau) \!+\! \bm{x}'\!, \bm{v}\big) \\ \vdots\\ 
    \min_{\tau \in [0,\!T]}  \min_{\bm{x}' \!\in\! \sR(\tau)}
    s_q\big(\phi(\tau) \!+\! \bm{x}'\!, \bm{v}\big) \\ 
    \Delta_{F,1}\big(\phi(T), \bm{v}\big) \\ 
    \vdots\\
    \Delta_{F,q_F}\big(\phi(T), \bm{v}\big)}\!\!.
  \end{equation}
  If $\Delta$ is continuously differentiable and $\sX \subset \R^n$ and $\sV
  \subset \R^{\ell}$ are compact, then there exists a minimum horizon
  $T_{\min} > 0$ such that for all $T \geq T_{\min}$, $\Delta$ is an RDSM.
\end{proposition}
\begin{proof}
  By compactness of $\tilde{\sC}_F$, and because $\Delta_F$ is continuously differentiable, define
  \begin{equation*}
    L_F = \max_{(\bm{x}', \bm{v}', i) \in \tilde{\sC}_{F} \times \{1,\ldots, q_F\}} \left\| D_1 \Delta_{F,i}(\bm{x}', \bm{v}') \right\| > 0.
  \end{equation*}
  By Assumption~\ref{ass:bounded-disturbance-function}, $\sup_{\bm{x} \in \R^n}\|p(\bm{x})\| \!\leq\! c_d/\sqrt{\lambda_{\min}(P)}$ is finite. By Assumption~\ref{ass:contractivity} and \cite[Proposition~1]{freire2026designing}, $\lim_{\tau \to \infty}  S_x(\tau) = 0$. Therefore, there exists $T_{\min} > 0$ such that $\forall T \geq T_{\min}$,
  \begin{equation*}
    \left\|S_{x}(T)\right\| \leq \frac{\epsilon \sqrt{\lambda_{\min}(P)}}{L_F    c_d d_{\max}}.
  \end{equation*}

  To prove property \eqref{eq:rdsm-safe}, let $(\bm{x}, \bm{v}) \in \tilde{\sC}$ be given and note that $\forall i \in \{1,\ldots, q\}$,
  \begin{align*}
    s_i(\bm{x}, \bm{v}) &= s_i\big(\phi(0), \bm{v}\big) \\
    &= \min_{\bm{x}' \in \sR(0)} s_i\big(\phi(0) + \bm{x}', \bm{v}\big) \\
    & \geq \Delta_i(\bm{x}, \bm{v}) \\
    &\geq 0,
  \end{align*}
  which implies $\bm{x} \in \sX_{\bm{v}}$. Next, note that $\big(\phi(T), \bm{v}\big) \in \tilde{\sC}_F \subset \tilde{\sX}$, which implies $\bm{v} \in \sV$.

  To prove property~\eqref{eq:rdsm-compact}, note that $\tilde{\sC}$ is bounded since $\tilde{\sC} \subset \tilde{\sX} \subset \sX \times \sV$.
  Thus, we only need to show that $\tilde{\sC}$ is closed in $\R^n \times \R^{\ell}$. This property follows by continuity of $\Delta$ and the fact that, by Assumption~\ref{ass:roa-large}, $\tilde{\sC} \subset \tilde{\sX} \subset \tilde{\sD}$, where $\tilde{\sX}$ is closed in $\R^n \times \R^{\ell}$.

  Property~\eqref{eq:rdsm-invariant} for $\Delta_i$ with $i \in \{1,\ldots,q\}$ holds by the same arguments made in the proof of Theorem~\ref{thm:traj-rdsm}.
  To prove property~\eqref{eq:rdsm-invariant} for $\Delta_{q+i}$ with $i \in \{1,\ldots, q_F\}$, pick $(\bm{x}, \bm{v}) \in \partial \tilde{\sC}$ and let $(i, \bm{d}) \in \tilde{\sI}(\bm{x}, \bm{v}) \times \sW$ be given. Recall that $\Delta_{q+1}(\bm{x}, \bm{v}) = \Delta_{F,i}\big(\phi(T), \bm{v}\big) = 0$, and therefore $\big(\phi(T), \bm{v}\big) \in \partial \tilde{\sC}_F$. Then,
  \begin{align*}
    &D_1 \Delta_{q+i}(\bm{x}, \bm{v}) \big(f_{\pi}(\bm{x},\bm{v}) + p(\bm{x}) \bm{d}\big)\\
    &\quad = D_1 \Delta_{F,i}\big(\phi(T), \bm{v}\big) D_2  \phi(T) \left(f_{\pi}(\bm{x}, \bm{v}) + p(\bm{x}) \bm{d}\right) \\
    &\quad = D_1 \Delta_{F,i}\big(\phi(T), \bm{v}\big) S_x(T) \left(f_{\pi}(\bm{x}, \bm{v}) + p(\bm{x}) \bm{d}\right)\\
    &\quad = D_1 \Delta_{F,i}\big(\phi(T), \bm{v}\big) \left(f_{\pi}\big(\phi(T), \bm{v}\big) +  S_x(T) p(\bm{x})\bm{d}\right)\\
    &\quad \geq \epsilon + D_1 \Delta_{F,i}\big(\phi(T), \bm{v}\big) S_x(T) p(\bm{x}) \bm{d} \\
    &\quad \geq \epsilon  - \left| D_1 \Delta_{F,i}\big(\phi(T), \bm{v}\big) S_x(T) p(\bm{x}) \bm{d} \right| \\
    &\quad \geq \epsilon - \left\| D_1 \Delta_{F,i}\big(\phi(T), \bm{v}\big) \right\| \|S_{x}(T)\| \; \|p(\bm{x})\| \; \|\bm{d}\| \\
    &\quad \geq \epsilon - \frac{L_F \|S_{x}(T)\| c_d d_{\max}}{\sqrt{\lambda_{\min}(P)}} \\
    &\quad \geq 0,
  \end{align*}
  where we have used \cite[Lemma~1]{freire2026designing}.
\end{proof}
A suitable option for the terminal RDSM $\Delta_F(\bm{x}, \bm{v})$ is the Lyapunov RDSM \eqref{eq:lyap-rdsm}.

\begin{remark}
  The conservativeness of the robust safety filtering is tied to the size of $\tilde{\sC} = \{(\bm{x},\bm{v}) \in \tilde{\sD} \mid \Delta(\bm{x}, \bm{v}) \geq 0\}$. For linear systems, $\tilde{\sC}_F \subset \tilde{\sC}$ because the contraction tube is nonconservative; this means filtering under $\Delta$ is no more conservative than filtering under $\Delta_F$. For conservative contraction tubes, enlargement of $\tilde{\sC}_F$ is not guaranteed.
\end{remark}

\subsection{Smoothness Considerations}
Next, we relax the assumption that $\Delta$ is continuously differentiable. To this end, we leverage families of functions. For any $\tau \in [0,\infty)$, define $\delta_{\tau}: \tilde{\sD} \to \R^q$, with channels
\begin{equation} \label{eq:traj-rdsm-family-tau}
  \delta_{\tau, i}(\bm{x}, \bm{v}) = \min_{\bm{x}' \in \sR(\tau)}
  s_i\big(\phi(\tau, \bm{x}, \bm{v}) + \bm{x}', \bm{v}\big).
\end{equation}
The superlevel set of $\delta_{\tau}$ is 
\begin{equation} \label{eq:set-superlevel-traj-dsm-family-tau}
  \tilde{\sC}_{\tau} = \left\{ (\bm{x}, \bm{v}) \in \tilde{\sD} \;\Big|\; \min_{i \in \{1,\ldots, q\}} \delta_{\tau, i}(\bm{x}, \bm{v}) \geq 0 \right\}.
\end{equation}
Given a terminal DSM $\Delta_F: \tilde{\sD} \to \R^{q_F}$ with horizon $T\geq 0$, define also $\delta_{F}: \tilde{\sD} \to \R^{q_F}$ by 
\begin{equation} \label{eq:traj-dsm-family-terminal}
  \delta_F(\bm{x}, \bm{v}) = \Delta_F\big(\phi(T,\bm{x}, \bm{v}), \bm{v}\big),
\end{equation}
with superlevel set
\begin{equation} \label{eq:set-superlevel-traj-dsm-family-terminal}
  \tilde{\sC}_F^T = \left\{(\bm{x}, \bm{v}) \in \tilde{\sD} \; \Big| \; \min_{i \in \{1,\ldots, q_F\}} \delta_{F, i}(\bm{x}, \bm{v}) \geq 0 \right\}.
\end{equation}
It can be shown that the superlevel set of \eqref{eq:traj-rdsm-fh} satisfies 
\begin{equation}
    \tilde{\sC} = \bigcap_{\tau \in [0,T]} \tilde{\sC}_{\tau} \cap \tilde{\sC}_F^T \subset \tilde{\sX}.
\end{equation}

To proceed, we require the joint map $(\tau, \bm{x}, \bm{v}) \mapsto \delta_{\tau}(\bm{x}, \bm{v})$ be continuously differentiable (see Remark~\ref{rem:ecs}). Such a requirement is not trivial due to the minimization over the contraction tube. To ensure the set of minimizers is a singleton, we make the following assumption about the constraints.
\begin{assumption}\label{ass:convex-cst}
  For any $\bm{v} \in \R^{\ell}$, and any $i\in \{1,\ldots,q\}$, the function  $s_i(\cdot, \bm{v})$ is either linear or strictly convex.
\end{assumption}
\begin{lemma}
  Under Assumption \ref{ass:convex-cst}, the joint map $\psi: (\tau, \bm{x}, \bm{v}) \mapsto \delta_{\tau}(\bm{x}, \bm{v})$ is continuously differentiable with
  \begin{subequations} \label{eq:traj-rdsm-family-tau-derivatives}
    \begin{align}
      D_1\delta_{\tau, i}(\bm{x}, \bm{v}) &= D_1s_i\big(\tilde{\phi}_i(\tau,
      \bm{x}, \bm{v}), \bm{v}\big) S_x(\tau),\\
      D_2\delta_{\tau, i}(\bm{x}, \bm{v}) &= D_1s_i\big(\tilde{\phi}_i(\tau,
      \bm{x}, \bm{v}), \bm{v}\big) S_v(\tau) \nonumber \\ 
      & \qquad \qquad + D_2s_i\big(\tilde{\phi}_i(\tau,
      \bm{x}, \bm{v}), \bm{v}\big),
    \end{align}
  \end{subequations}
  where $\tilde{\phi}_i:[0,\infty) \times \tilde{\sD}$ is the worst-case (with respect to $s_i$) disturbed trajectory at $\tau$, starting from $\bm{x}$ under constant reference $\bm{v}$, defined as
  \begin{equation} \label{eq:traj-rdsm-family-tau-worst-case}
    \tilde{\phi}_i(\tau, \bm{x}, \bm{v}) = \underset{\bm{z} \in \bar{\sR}}{\argmin} \; s_i\big(\phi(\tau, \bm{x}, \bm{v}) + (1 - e^{-c\tau}) \bm{z}, \bm{v}\big).
  \end{equation}
\end{lemma}
\begin{proof}
  Recalling the contraction tube $\sR$ is continuous by Lemma~\ref{lem:contraction-tube-continuous}, the joint map $\psi$ is continuous by the Berge maximum theorem \cite{berge1963topological}

  Recall also the steady-state contraction tube $\bar{\sR} \subset \R^n$, defined in \eqref{eq:contraction-tube-ss}, is compact. Define the function $\xi_i:[0,\infty) \times \tilde{\sD} \times \bar{\sR} \to \R$ by
  \begin{equation*}
    \xi_i(\tau, \bm{x}, \bm{v}, \bm{z}) = s_i\big(\phi(\tau, \bm{x}, \bm{v})  + (1 - e^{-c\tau})\bm{z}, \bm{v}\big).
  \end{equation*}
  Note that $\xi_i$ is continuously differentiable. Then,
  \begin{equation*}
    \psi_i(\tau, \bm{x}, \bm{v}) = \min_{\bm{z} \in \bar{\sR}} \, \xi_i(\tau,
    \bm{x}, \bm{v}, \bm{z}).
  \end{equation*}
  Define the set of minimizing points
  \begin{equation*}
    \bar{\sR}_i(\tau, \bm{x}, \bm{v}) \!=\! \left\{\bar{\bm{z}} \!\in\! \bar{\sR} \mid
    \xi_i(\tau, \bm{x}, \bm{v}, \bar{\bm{z}}) \!=\! \min_{\bm{z} \in \bar{\sR}}
    \xi_i(\tau, \bm{x}, \bm{v}, \bm{z})\right\},
  \end{equation*}
  and note that by Assumption \ref{ass:convex-cst} and convexity of $\bar{\sR}$, $\bar{\sR}_i(\tau, \bm{x}, \bm{v})$ consists of a unique point $\bar{\bm{z}}$. It follows by the Danskin theorem \cite[Proposition~B.22]{bertsekas2016nonlinear} that $\psi$ is differentiable at $(\tau, \bm{x}, \bm{v})$ with the derivatives given in the statement of this Lemma.
  Leveraging again the Berge maximum theorem \cite{berge1963topological}, the solution  mapping $\bar{z}:(\tau, \bm{x}, \bm{v}) \mapsto \bar{\bm{z}} \in \bar{\sR}_i(\tau, \bm{x}, \bm{v})$ is continuous. Therefore, the joint map $\psi$ is continuously differentiable.
\end{proof}
Note that Assumption~\ref{ass:convex-cst} only ensures there is a
unique minimizer in $\bar{\sR}$. We have shown that instead of considering the monolithic $\Delta$ given in \eqref{eq:traj-rdsm-fh}, we can instead, and equivalently, work with the family of continuously differentiable functions $\{\delta_{\tau}\}_{\tau \in [0,T]} \cup \delta_F$ (see Remark~\ref{rem:ecs}).

\begin{theorem} \label{thm:traj-rdsm-family-ercs}
  Let $\Delta_F: \tilde{\sD} \to \R^{q_F}$ be an RDSM satisfying \eqref{eq:rdsm-invariant-strict} and let $T \geq 0$ be the horizon. For any $\tau \in [0,T]$, let $\delta_\tau: \tilde{\sD} \to \R^q$ be given as in \eqref{eq:traj-rdsm-family-tau} and $\delta_F: \tilde{\sD} \to \R^{q_F}$ be given as in \eqref{eq:traj-dsm-family-terminal}. If $\sX \subset \R^n$ and $\sV \subset \R^{\ell}$ are compact, there exists a minimum horizon $T_{\min} > 0$ such that if $T \geq T_{\min}$, then the family $\{\delta_{\tau}\} \cup \delta_F$ has the extended robust control-sharing property.
\end{theorem}
\begin{proof}
  Since $\Delta_F$ is an RDSM-RCBF by Theorem \ref{thm:rdsm-rcbf}, it follows by Corollary \ref{cor:rdsm-rcbf-cond-pi} that there exists  $\alpha_F \in \cK$ satisfying \eqref{eq:rdsm-rcbf-cond-pi}.
  Moreover, by compactness of $\tilde{\sC}_F$, and because $\Delta_F$ is continuously differentiable, define
  \begin{equation*}
    L_F = \max_{(\bm{x}', \bm{v}', i) \in \tilde{\sC}_F\times \{1,\ldots, q_F\}} \left\|D_1\Delta_{F,i}(\bm{x}', \bm{v}')\right\| > 0.
  \end{equation*}
  By Assumption~\ref{ass:bounded-disturbance-function}, $\sup_{\bm{x} \in \R^n} \|p(\bm{x})\| \!\leq\! c_d/ \sqrt{\lambda_{\min}(P)}$ is finite. By Assumption~\ref{ass:contractivity} and \cite[Proposition~1]{freire2026designing}, $\lim_{\tau \to \infty} S_x(\tau) = 0$.
  Therefore, there exists $T_{\min} > 0$ such that if $T \geq T_{\min}$,
  \begin{equation*}
    \left\|S_{x}(T)\right\| \leq \frac{\epsilon \sqrt{\lambda_{\min}(P)}}{L_F c_d d_{\max}}.
  \end{equation*}

  Recall that if $\sX$ and $\sV$ are compact, so is $\tilde{\sC}$ (see, e.g., the proof of Theorem \ref{thm:traj-rdsm}). Pick any $i \in \{1,\ldots,q\}$ and define the set-valued map $\hat{\sC}_i: [0,\infty) \times [0,T] \rightrightarrows \tilde{\sD}$ by
  \begin{equation*}
    \hat{\sC}_i(c, \tau) = \left\{(\bm{x}, \bm{v}) \in \tilde{\sC} \mid 0 \leq \delta_{\tau, i}(\bm{x}, \bm{v}) \leq c\right\},
  \end{equation*}
  the function $\tilde{\xi}_{i}: \tilde{\sD} \times \sW \times [0,T] \to \R$ by
  \begin{equation*}
    \tilde{\xi}_i (\bm{x}, \bm{v}, \bm{d}, \tau) = D_1 \delta_{\tau, i}(\bm{x}, \bm{v}) \big(f_{\pi}(\bm{x}, \bm{v}) + p(\bm{x})\bm{d} \big),
  \end{equation*}
  and the value function $\hat{\alpha}_i: [0,\infty) \times [0,T] \to \R$ by
  \begin{equation*}
    \hat{\alpha}_i(c, \tau) = - \inf_{(\bm{x}, \bm{v}, \bm{d}) \in \hat{\sC}_i(c, \tau) \times \sW} \tilde{\xi}_i(\bm{x}, \bm{v}, \bm{d}, \tau).
  \end{equation*}
  We seek to upper-bound $\hat{\alpha}_i$ by a class $\cK$ function as in the proof of Theorem~\ref{thm:rcbf-rci}.

  For the first condition ($\hat{\alpha}_i$ upper semicontinuous), it can be shown that $c \mapsto \hat{\sC}_i(c) \times \sW$ is upper semicontinuous. Recalling that $c \mapsto \hat{\sC}_i(c) \times \sW$ has compact values, it follows by the Berge maximum theorem \cite{berge1963topological} that the value function $\hat{\alpha}_i$ is upper semicontinuous on $[0,\infty) \times [0,T]$. Next, define
  \begin{equation*}
    \bar{\alpha}_i(c) = \sup_{\tau \in [0,T]} \hat{\alpha}_i(c,\tau),
  \end{equation*}
  and note that we can apply the Berge maximum theorem \cite{berge1963topological} again
  to conclude that $\bar{\alpha}_i$ is upper semicontinuous on $[0,\infty)$.

  For the second condition ($\hat{\alpha}(0) \leq 0$), let $\tau \in [0,T)$ be given and pick any $(\bm{x}, \bm{v}, \bm{d}) \in \hat{\sC}_i(0, \tau) \times \sW \subset  \tilde{\sC} \times \sW$. For a contradiction, assume $\tilde{\xi}_i(\bm{x}, \bm{v}, \bm{d}, \tau) < 0$ and let $\tilde{x}:[0,\infty) \to \R^n$ be the solution to the initial value problem
  \begin{equation*}
    D\tilde{x}(t) = f_{\pi}\big(\tilde{x}(t), \bm{v}\big) + p\big(\tilde{x}(t)\big) \bm{d}, \quad \tilde{x}(0) = \bm{x}.
  \end{equation*}
  Then there exists $\epsilon \in (0,T-\tau)$ such that
  \begin{align*}
    0 &> \delta_{\tau, i}\big(\tilde{x}(\epsilon), \bm{v}\big) \\ 
    &> \min_{\bm{x}' \in \sR(\tau)} s_i\Big( \phi\big(\tau, \tilde{x}(\epsilon)\big) + \bm{x}', \bm{v}\Big)\\
    &> \min_{\bm{x}' \in \sR(\tau + \epsilon)} s_i\big(\phi(\tau + \epsilon, \bm{x}) + \bm{x}', \bm{v}\big) \\
    &> \delta_{\tau + \epsilon, i}(\bm{x}, \bm{v}),
  \end{align*}
  which contradicts $(\bm{x}, \bm{v}) \in \tilde{\sC}$. In the previous string of inequalities, we have used the fact that (see proof of Theorem~\ref{thm:traj-rdsm}) for any $\sigma \geq 0$,
  \begin{equation*}
    \phi\big(\sigma, \tilde{x}(\epsilon), \bm{v}\big) + \sR(\sigma) \subset \phi(\sigma + \epsilon, \bm{x}, \bm{v}) + \sR(\sigma + \epsilon).
  \end{equation*}
  By joint continuity of the map $(\tau, \bm{x}, \bm{v}) \mapsto \delta_{\tau, i}(\bm{x}, \bm{v})$, we have shown that $\forall \tau \in [0,T]$, $\tilde{\xi}_i(\bm{x}, \bm{v}, \bm{d}, \tau) \geq 0$, which implies $\hat{\alpha}_i(0, \tau) \leq 0$. It follows that $\bar{\alpha}_i(0) \leq 0$.

  Having shown its conditions hold, it follows by \cite[Lemma~4.1]{albertini1999continuous} that there exists $\alpha_i \in \cK$ that upper-bounds $\bar{\alpha}_i$. Since $i \in \{1,\ldots,q\}$ was arbitrary, let us define $\alpha(c) = \max_{i \in \{1,\ldots,q,F\}} \alpha_i(c)$ and note that $\alpha \in \cK$. 
  
  Finally, let $(\bm{x}, \bm{v}) \in \tilde{\sC}$ be given and recall that since $\tilde{\sC} \subset \tilde{\sX}$, $\pi(\bm{x}, \bm{v}) \in \sU$. For any $\tau \in [0,T]$, any $i \in \{1,\ldots,q\}$, and any $\bm{d} \in \sW$,
  \begin{align*}
    \tilde{\xi}_i(\bm{x}, \bm{v}, \bm{d}, \tau) \! &= D\delta_{\tau, i}(\bm{x}, \bm{v}) \mat{f(\bm{x}) + g(\bm{x}) \pi(\bm{x}, \bm{v}) + p(\bm{x}) \bm{d} \\ 0} \\
    &\geq \inf_{(\bm{x}_0, \bm{v}_0, \bm{d}_0) \in \hat{\sC}_i\left(\delta_{\tau, i}(\bm{x}, \bm{v}), \tau\right) \times \sW} \tilde{\xi}_i(\bm{x}_0, \bm{v}_0, \bm{d}_0, \tau) \\
    &\geq -\hat{\alpha}_i\big(\delta_{\tau, i}(\bm{x}, \bm{v}), \tau\big) \\
    &\geq -\bar{\alpha}_i\big(\delta_{\tau, i}(\bm{x}, \bm{v})\big)\\
    &\geq -\alpha\big(\delta_{\tau,i}(\bm{x}, \bm{v})\big).
  \end{align*}
  Also, since $(\bm{x}, \bm{v}) \in \tilde{\sC}$, it follows that
  \begin{equation*}
    0 \leq \min_{i \in \{1,\ldots,q_F\}} \delta_{F,i}(\bm{x}, \bm{v}) = \min_{i \in \{1,\ldots, q_F\}} \Delta_{F,i}\big(\phi(T, \bm{x}, \bm{v}), \bm{v}\big),
  \end{equation*}
  which implies $\big(\phi(T), \bm{v}\big) \in \tilde{\sC}_F$. Therefore, recalling $\alpha_F$ satisfies \eqref{eq:rdsm-rcbf-cond-pi},
  \begin{align*}
    D\delta_{F, i}&(\bm{x}, \bm{v}) \mat{f(\bm{x}) + g(\bm{x}) \pi(\bm{x},    \bm{v}) + p(\bm{x})\bm{d} \\ 0} \\ 
    &= D_1\delta_{\tau,i}(\bm{x}, \bm{v})
    \big(f_{\pi}(\bm{x}, \bm{v}) + p(\bm{x}) \bm{d} \big) \\
    &= D_1 \Delta_{F,i}\big(\phi(T), \bm{v}\big) S_x(T) \big( f_{\pi}(\bm{x},
    \bm{v}) + p(\bm{x}) \bm{d} \big)\\
    &= D_1 \Delta_{F,i}\big(\phi(T), \bm{v}\big)\big( f_{\pi}(\phi(T), \bm{v})
    + S_x(T) p(\bm{x}) \bm{d}\big)\\
    &\geq \epsilon + D_1\Delta_{F,i} \big(\phi(T), \bm{v}\big) S_x(T)
    p(\bm{x}) \bm{d}\\
    &\geq \epsilon - \left|D_1\Delta_{F,i} \big(\phi(T), \bm{v}\big) S_x(T)
    p(\bm{x}) \bm{d}\right|\\
    &\geq \epsilon - \left\|D_1\Delta_{F,i} \big(\phi(T),
    \bm{v}\big)\right\|\; \|S_x(T)\| \; \left\| p(\bm{x})\right\|
    \|\bm{d}\|\\
    &\geq \epsilon - \frac{L_F \|S_x(T)\| c_d
    d_{\max}}{\sqrt{\lambda_{\min}(P)}}\\
    &\geq 0\\
    & \geq -\alpha_F\Big(\Delta_{F,i}\big(\phi(T), \bm{v}\big)\Big) \\
    & \geq -\alpha\big(\delta_{F,i}(\bm{x}, \bm{v})\big),
  \end{align*}
  where we have used \cite[Lemma 1]{freire2026designing}. We conclude that the family $\{\delta_{\tau}\} \cup \delta_F$ has the extended robust control-sharing property.
\end{proof}
Note a family of functions with the extended robust control-sharing property can be treated as the generalization of the CBF definition to (infinite) families of continuously differentiable functions (see Remark~\ref{rem:ecs}).

\subsection{Trajectory RDSM Safety Filter}
The RDSM-RCBF safety filter for the family $\{\delta_\tau\} \cup \delta_F$ given in \eqref{eq:traj-rdsm-family-tau} and \eqref{eq:traj-dsm-family-terminal} is
\begin{align} \label{eq:rdsm-rcbf-sf-family}
  \min_{(\bm{u}, \bm{w}) \in \sU \times \R^{\ell}} \; &\|\bm{u} - \kappa(\bm{x}, \bm{v})\|^2 + \eta \|\bm{w} - \rho(\bm{x}, \bm{v})\|^2\\
  \text{s.t.} \quad & \dot{\delta}_{\tau,i}(\bm{x},\bm{v}, \bm{u}, \bm{w}) \!+\! \tilde{\omega}_{\tau, i}(\bm{x}, \bm{v}) \geq -\alpha\big(\delta_{\tau, i}(\bm{x}, \bm{v})\big),\nonumber\\
  & \qquad\qquad\qquad\qquad \forall \tau \in [0,T], \; i = 1,\ldots,q \nonumber\\
  &\dot{\delta}_{F,i}(\bm{x}, \bm{v}, \bm{u}, \bm{w}) \!+\! \tilde{\omega}_{F,i}(\bm{x}, \bm{v}) \geq - \alpha\big(\delta_{F,i}(\bm{x}, \bm{v})\big), \nonumber\\
  &\qquad \qquad \qquad \qquad \qquad \quad \;\qquad i = 1, \ldots, q_F, \nonumber
\end{align}
where
\begin{equation*}
  \dot{\delta}_{\tau, i}(\bm{x}, \bm{v}, \bm{u}, \bm{w}) =
  D \delta_{\tau,i}(\bm{x}, \bm{v}) \mat{f(\bm{x}) + g(\bm{x}) \bm{u} \\
  \bm{w}},
\end{equation*}
with $D \delta_{\tau, i}$ given in
\eqref{eq:traj-rdsm-family-tau-derivatives}, where $\dot{\delta}_{F,i}$ is given as
\begin{equation*}
  \dot{\delta}_{F, i}(\bm{x}, \bm{v}, \bm{u}, \bm{w}) = D \delta_{F,i}(\bm{x}, \bm{v}) \mat{f(\bm{x}) + g(\bm{x}) \bm{u} \\ \bm{w}},
\end{equation*}
with
\begin{align*}
  D_1 \delta_{F,i}(\bm{x}, \bm{v}) \!&=\! D_1 \Delta_{F,i}\big(\phi(T),
  \bm{v}\big) S_x(T) \\
  D_2 \delta_{F,i}(\bm{x}, \bm{v}) \!&=\! D_1 \Delta_{F,i}\big(\phi(T), \bm{v}\big) S_v(T) 
  \!+\! D_2 \Delta_{F,i}\big(\phi(T), \bm{v}\big),
\end{align*}
and where
\begin{align*}
  \tilde{\omega}_{\tau,i}(\bm{x}, \bm{v}) &= \min_{\bm{d} \in \sW}
  D_1 \delta_{\tau, i}(\bm{x}, \bm{v}) p(\bm{x}) \bm{d},\\
  \omega_F(\bm{x}, \bm{v}) &= \min_{\bm{d} \in \sW} D_1 \delta_{F,i}(\bm{x},
  \bm{v}) p(\bm{x}) \bm{d}.
\end{align*}
While \eqref{eq:rdsm-rcbf-sf-family} has infinite-dimensional constraints, we consider in practice a finite set $\sT \subset [0,T]$ and enforce the constraints $\forall \tau \in \sT$. The set $\sT$ can be obtained, along with $\phi$, using a variable-step solver for \eqref{eq:ivp-prestab}. This can be done to an arbitrary degree of accuracy. Similar to \eqref{eq:rdsm-rcbf-sf}, the trajectory RDSM-RCBF safety filter \eqref{eq:rdsm-rcbf-sf-family} is a QP if $\sU$ is a polytope.
Also important to consider is the computation cost of the worst-case (with respect to $s_i$) disturbed trajectory $\tilde{\phi}_i$ given in \eqref{eq:traj-rdsm-family-tau-worst-case} and $\tilde{\omega}_{\tau,i}$.
In both cases, Assumptions~\ref{ass:compact-convex-disturbance-set} and~\ref{ass:convex-cst} ensure they can be computed easily. Feasibility of \eqref{eq:rdsm-rcbf-sf-family} at any $(\bm{x}, \bm{v}) \in \tilde{\sC}$ follows by Theorem~\ref{thm:traj-rdsm-family-ercs}.

\section{RDSM-RCBFs for Linear Systems} \label{sec:rdsm-rcbf-linear}
In this section, we show how the design of RDSMs simplifies for linear systems subject to disturbances
\begin{equation} \label{eq:sys-disturbance-linear}
  \dot{\bm{x}} = A\bm{x} + B\bm{u} + E\bm{d},
\end{equation}
where $A \in \R^{n \times n}$, $B \in \R^{n \times m}$ and $E \in \R^{n \times
n_d}$. We also assume the constraints are linear.
\begin{equation*}
  \sX = \left\{\bm{x} \in \R^n \mid  W\bm{x} \leq \bm{z}\right\}, \quad \sU = \left\{\bm{u} \in \R^m \mid  M\bm{u} \leq \bm{b}\right\},
\end{equation*}
where $\bm{z} \in \R^{n_x}$ and $\bm{b} \in \R^{n_u}$, and define $q = n_x + n_u$.
If we assume the pair $(A,B)$ is stabilizable, there exists $K \in \R^{m \times n}$ such that $A_{\pi} = A - BK$ is Hurwitz. Letting $G = [G_x;\, G_u]$, with $G_x \in \R^{n\times {\ell}}$ and $G_u \in \R^{m\times {\ell}}$ be a basis for $\ker{[A \; B]}$,
\begin{equation}
  \pi(\bm{x}, \bm{v}) = G_u \bm{v} - K \left(\bm{x} - G_x\bm{v}\right),
\end{equation}
is a prestabilizing controller satisfying $\pi(G_x \bm{v}, \bm{v}) = G_u
\bm{v}$. Then, the prestabilized constraint set $\sX_{\bm{v}}$ defined in \eqref{eq:set-prestab-set} is described by
\begin{equation}
  s(\bm{x}, \bm{v}) = \mat{\bm{z} - W\bm{x} \\ \bm{b} - M\pi(\bm{x}, \bm{v})},
\end{equation}
which is linear and therefore continuously differentiable with
\begin{equation*}
  D_1s = \mat{-W\\ MK} \!\in\! \R^{q \times n},\;
  D_2s = \mat{0 \\ M(G_u + K G_x)} \!\in\! \R^{q \times m},
\end{equation*}
where we have omitted the arguments of $D_1s$ and $D_2s$ because they are constant.
To construct the RDSMs, we can rely on any $P \succ 0$ satisfying
\begin{equation}
  PA_{\pi} + A_{\pi}^\top P \prec 0,
\end{equation}
which exists by the Lyapunov theorem because $A_{\pi}$ is Hurwitz.

\begin{remark}
  A common approach to construct $K$ and $P$ is the linear--quadratic
  regulator (LQR). Given LQR weights $Q \succeq 0$ and $R \succ 0$, with
  $(A,Q)$ detectable, the LQR optimal gain is $K = R^{-1} B^\top P$,  where $P \succ 0$ is the unique solution to the algebraic Riccati equation
  $0 = PA + A^\top P + Q - PBR^{-1} B^\top P$.
\end{remark}

\subsection{Lyapunov Robust Dynamic Safety Margin} \label{ssec:lyap-rdsm}
Constructing a Lyapunov RDSM amounts to finding an ISS Lyapunov function. To
that end, let
\begin{equation*}
  V(\bm{x}, \bm{v}) = \left(\bm{x} - G_x\bm{v}\right)^\top P \left(\bm{x} -
  G_x\bm{v}\right).
\end{equation*}
The function $V$ satisfies \eqref{eq:lyap-iss-bounded} because $P \succ 0$.
Indeed,
\begin{equation*}
  \lambda_{\min}(P)\|\bm{x} - G_x \bm{v}\|^2 \leq V(\bm{x}, \bm{v}) \leq
  \lambda_{\max}(P) \|\bm{x} - G_x \bm{v}\|^2.
\end{equation*}
Additionally, letting $B_{\pi} = B(G_u + K_{\pi}G_x)$, we have that the
prestabilized dynamics satisfy $f_{\pi}(\bm{x}, \bm{v}) = A_{\pi} \bm{x} + B_{\pi} \bm{v}$.
Recall that $D_1V(\bm{x}, \bm{v}) = 2\left(\bm{x} - G_x\bm{v}\right)^\top
P$ and that $A_{\pi}G_x \bm{v} + B_{\pi}\bm{v} = 0$. With this,
\begin{equation*}
  D_1V(\bm{x}, \bm{v}) f_{\pi}(\bm{x}, \bm{v}) = \bm{e}^\top \left(P A_{\pi} + A_{\pi}^\top P \right) \bm{e} = - \bm{e}^\top Q_{\pi} \bm{e},
\end{equation*}
where $\bm{e} = \bm{x} - G_x\bm{v}$ and $Q_{\pi} = - P A_{\pi} - A_{\pi}^\top
P \succ 0$. Adding the disturbance term,
\begin{equation*}
  D_1 V(\bm{x}, \bm{v}) \big(f_{\pi}(\bm{x}, \bm{v}) + E\bm{d} \big) =
  -\bm{e}^\top Q_{\pi} \bm{e} + 2\bm{e}^\top P E\bm{d}.
\end{equation*}
Defining the vector $\bm{z} = Q_{\pi}^{1/2} \bm{e}$, we can bound this term by
\begin{align*}
  D_1 V(\bm{x}, \bm{v}) \big(f_{\pi}(\bm{x}, \bm{v}) + E\bm{d} \big) &= -
  \|\bm{z}\|^2  + 2 \bm{z}^\top Q_{\pi}^{-1/2} P E\bm{d} \\
  & \!\!\!\!\!\!\!\!\!\!\!\!\!\!\!\!\!\!\!\!\!\!\!\! \leq -
  \|\bm{z}\| \left(\|\bm{z}\| - 2\left\|Q_{\pi}^{-1/2} P E\right\|
  \|\bm{d}\|\right).
\end{align*}
Note that if
\begin{equation*}
  \|\bm{z}\| \geq 2\left\|Q_{\pi}^{-1/2} P E\right\| \|\bm{d}\|,
\end{equation*}
then $D_1V(\bm{x}, \bm{v})\big(f_{\pi}(\bm{x}, \bm{v}) + E\bm{d}\big) \leq 0$.
We can guarantee this condition when $V(\bm{x},\bm{v}) \geq \chi(\|\bm{d}\|)$ with
\begin{equation*}
   \chi(\|\bm{d}\|) \triangleq 4
  \lambda_{\max}\left(Q_{\pi}^{-1/2} P Q_{\pi}^{-1/2}\right)
  \left\|Q_{\pi}^{-1/2} P E\right\|^2 \|\bm{d}\|^2.
\end{equation*}
We have shown $V$ is a reference-dependent ISS Lyapunov function as in
Definition~\ref{def:lyap-iss}. For linear systems in the form of
\eqref{eq:sys-disturbance-linear}, the region of attraction is global
(i.e., $\sD_{\bm{v}} = \R^n$), and the $i$-th ($i \in \{1,\ldots, q\}$) safety threshold value is
\begin{equation}
  \Gamma^*_i(\bm{v}) = \frac{s_i\left(G_x\bm{v},
  \bm{v}\right)\left|s_i\left(G_x\bm{v}, \bm{v}\right)\right|}{D_1s_i P^{-1}
  D_1 s_i^\top}.
\end{equation}
Finally, the Lyapunov RDSM $\Delta:\R^n \times \R^{m} \to \R^{2q}$ is
\begin{equation}
  \Delta(\bm{x}, \bm{v}) = \mat{\Gamma_1^*(\bm{v}) - V(\bm{x}, \bm{v}) \\
  \Gamma^*_1(\bm{v}) - \chi(d_{\max}) \\ \vdots \\ \Gamma_q^*(\bm{v}) -
  V(\bm{x}, \bm{v}) \\ \Gamma_q^*(\bm{v}) - \chi(d_{\max})}.
\end{equation}

\subsection{Trajectory Robust Dynamic Safety Margin} \label{ssec:traj-rdsm}
First, we verify that the prestabilized system $f_{\pi}$ satisfies the
contraction assumptions with
\begin{equation}
  c = -\mu_P(A_{\pi}) > 0, \qquad   c_d = \left\| P^{1/2} E\right\| > 0.
\end{equation}

\begin{remark}
  Let $T$ be a nonsingular transformation such that $TA_{\pi} T^{-1} = J$,
  where $J$ is the Jordan normal form of $A_{\pi}$. The choice $P = T^* T
  \succ 0$, satisfies the Lyapunov condition $PA_{\pi} + A_{\pi}^\top P \prec
  0$ and achieves the largest contraction rate $c$. Refer to
  \cite{bullo2026contraction} for more details on the efficiency of induced
  norms. While this choice is usually desirable for stability analysis, we
  have found that the geometry of the contraction tube (dictated by the
  eigenvalues of $P$) and its alignment with respect to the constraints
  $\sX_{\bm{v}}$ is more important than the magnitude of the contraction rate
  $c$.
\end{remark}
Assumption~\ref{ass:contractivity} holds because, for any $\bm{x}_1,
\bm{x}_2 \in \R^n$, $\bm{v} \in \R^{m}$, and $\bm{d} \in \R^{n_d}$,
\begin{align*}
  (\dot{\bm{x}}_1 &- \dot{\bm{x}}_2)^\top P (\bm{x}_1 - \bm{x}_2)  \\
  &=
  (\bm{x}_1 - \bm{x}_2)^\top A_{\pi}^\top P (\bm{x}_1 - \bm{x}_2) \\
  &= \frac12 (\bm{x}_1 - \bm{x}_2)^\top \left( PA_{\pi} +  A_{\pi}^\top P
  \right) (\bm{x}_1 - \bm{x}_2)\\
  &\leq \frac{1}{2}\lambda_{\max}\left(PA_{\pi}P^{-1} + A^\top \right) \|\bm{x}_1 - \bm{x}_2\|_P^2\\
  &\leq  \mu_{P}(A_{\pi}) \|\bm{x}_1 - \bm{x}_2\|_P^2 \\
  &\leq -c\|\bm{x}_1 - \bm{x}_2 \|_P^2.
\end{align*}

Assumption~\ref{ass:bounded-disturbance-function} holds because, for any
$\bm{d}_1, \bm{d}_2 \in \R^{n_d}$,
\begin{align*}
  \|E(\bm{d}_1 - \bm{d}_2)\|_P &= \left\|P^{1/2}E(\bm{d}_1 - \bm{d}_2)\right\|
  \\ 
  &\leq \left\|P^{1/2} E\right\| \|\bm{d}_1 - \bm{d}_2\| \\
  &\leq c_d  \|\bm{d}_1 -
  \bm{d}_2\|.
\end{align*}
One of the main pieces of computing the derivatives of the family
$\{\delta_{\tau}\}$ is the worst-case (with respect to $s_i$) disturbed
trajectory at $\tau$ starting from $\bm{x}$ under constant reference $\bm{v}$,
denoted $\tilde{\phi}_i(\tau, \bm{x}, \bm{v})$ and defined in
\eqref{eq:traj-rdsm-family-tau-worst-case}. From the KKT conditions, we
obtain that
\begin{equation*}
  \tilde{\phi}_i(\tau) = \phi(\tau) -
  \frac{c_d(1 - e^{-c\tau})d_{\max}}{c} \frac{P^{-1}
  D_1s_i^\top}{\|D_1s_i^\top\|_{P^{-1}}}.
\end{equation*}

\section{Simulation Examples}  \label{sec:examples}
We present two simulation examples applying the Lyapunov and trajectory RDSM-RCBFs. We compare their performance with a nonrobust implementation of the trajectory DSM-CBF in \cite{freire2026designing}. Simulations are computed in MATLAB and optimization problems are solved with MOSEK \cite{mosek}.

\subsection{Double Integrator}
Consider a double integrator system subject to disturbances $\ddot{x} = u + d$ with system matrices
\begin{equation}
  A = \mat{0 & 1 \\ 0 & 0}, \quad B = \mat{0\\ 1}, \quad E = \mat{0 \\ 1}.
\end{equation}
The equilibrium manifold is described by $\bar{x}(v) = v$ and $\bar{u}(v) = 0$. 
The state and input constraint sets are
\begin{align*}
  \sX &= \left\{\bm{x} \in \R^2 \mid |x| \leq 0.8~\unit{\m}, \, |\dot{x}| \leq 1.5~\unit{m/s}\right\},\\
  \sU &= \left[-20~\unit{m/s^2}, 20~\unit{m/s^2}\right],
\end{align*}
and the disturbance set is $\sW = [-2~\unit{m/s^2}, 2~\unit{m/s^2}]$. The initial state is $\bm{x}_0 = [0;\, 0] \in \sX$ and the nominal controller is $\kappa(\bm{x}) = -31.6(x-1) - 12.8\dot{x}$, which stabilizes the system to the inadmissible equilibrium $[1;\, 0]$. To design the gain $K_{\pi}$ for the prestabilizing controller, we solve the LQR problem with weights $Q = I_2$ and $R =0.01$. This results in $\pi(\bm{x}, v) = -10(x - v) -11 \dot{x}$ and associated weighting matrix
\begin{equation*}
  P = \mat{1.1 & 0.1 \\ 0.1 & 0.11} \succ 0.
\end{equation*}
We implement the following three approaches:
\begin{enumerate}
  \item \textbf{Trajectory DSM-CBF} \cite{freire2026designing}: This approach is not robust to external disturbances, which causes the DSM-CBF safety filter to become infeasible during the simulation. To address this, we add a slack variable to relax the constraints. The terminal DSM is a Lyapunov DSM computed with $\pi$ as described in \cite{freire2026using}.
  \item \textbf{Lyapunov RDSM-RCBF:} We follow the approach in Section~\ref{ssec:lyap-rdsm} to design a Lyapunov RDSM and implement the RDSM-RCBF safety filter \eqref{eq:rdsm-rcbf-sf}.
  \item \textbf{Trajectory RDSM-RCBF:} 
  We follow the approach in Section~\ref{ssec:traj-rdsm} to compute the contraction rate $c$, the disturbance rate $c_d$, and the worst-case disturbed trajectory $\tilde{\phi}$. We implement the RDSM-RCBF safety filter given in \eqref{eq:rdsm-rcbf-sf-family} for the families $\{\delta_{\tau}\}$ and $\{\delta_F\}$ where the terminal RDSM-RCBF is exactly the Lyapunov RDSM-RCBF computed in the previous approach.
\end{enumerate}
The prediction horizon in all trajectories that require it is $T =
0.5~\unit{s}$. All class $\cK$ functions used are $\alpha:c \mapsto 50c$.

\begin{figure}[t]
\centering
  \includegraphics[width=0.90\linewidth]{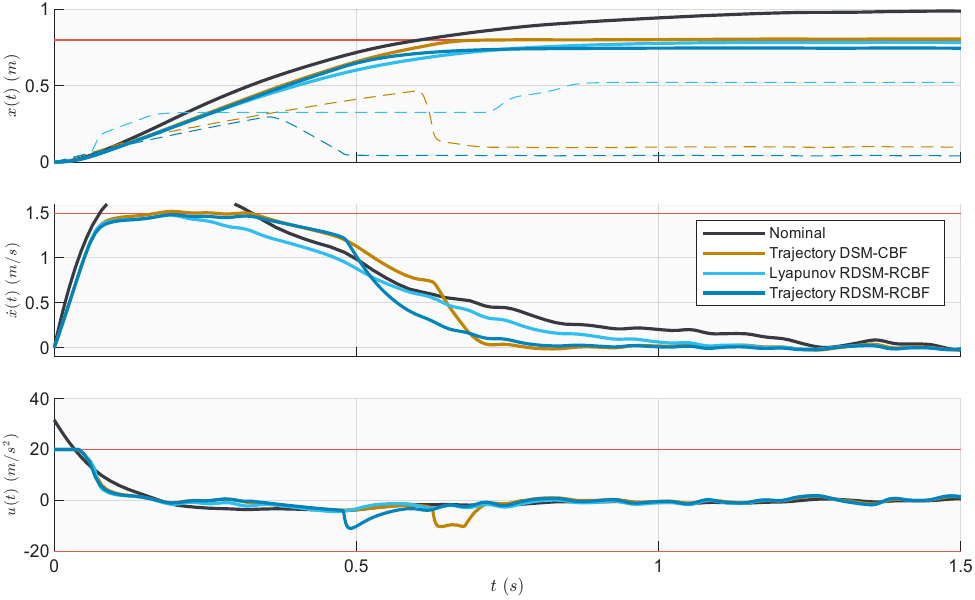}
  \includegraphics[width=0.90\linewidth]{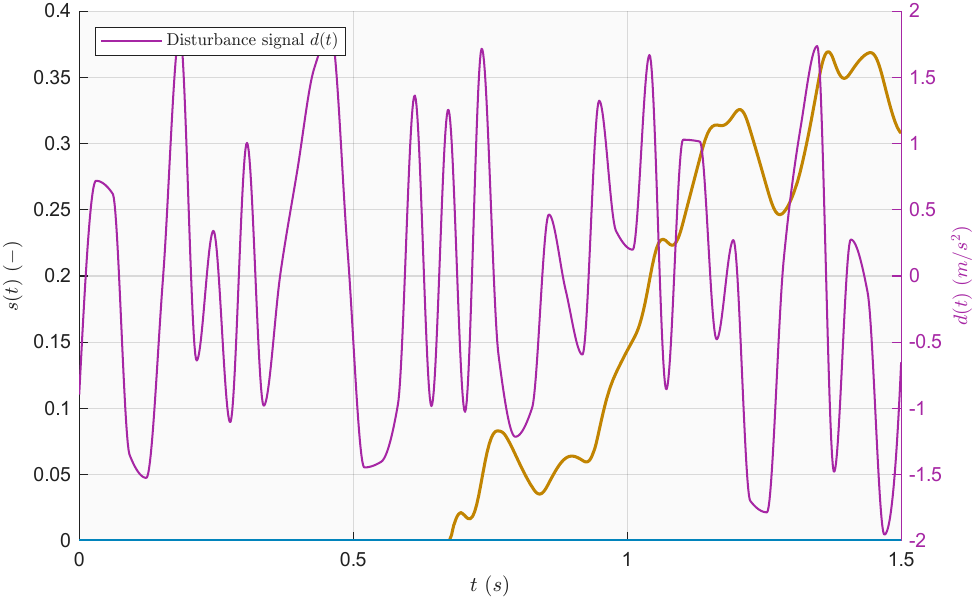}
  \caption{Simulation of the double integrator model subject to disturbances.
  Both robust approaches successfully enforce all the constraints while the
  trajectory DSM-CBF fails to do so. In addition, the trajectory DSM-CBF
  required augmenting the DSM-CBF safety filter with a slack variable to
  prevent infeasibility, its value is given in the bottom plot, which also
  shows that the robust approaches did not require the use of their slack
  variable to maintain feasibility, as expected.}
  \label{fig:dint}
\end{figure}

Fig.~\ref{fig:dint} shows the response of each approach along with
the value of the slack variables. The trajectory DSM-CBF performs best, but
violates (slightly) the constraints and also required the slack variable to
maintain feasibility. The Lyapunov RDSM-RCBF respects all constraints without
the need for the slack variable.  Its performance is slower than the
trajectory DSM-CBF, but remains competitive nonetheless. It also has the
lowest computational footprint because it does not require trajectory
predictions. The trajectory RDSM-RCBF also respects all constraints without
the need for the slack variable. However, the performance is mixed, exhibiting
faster rise time than the Lyapunov RDSM-RCBF, similar to the trajectory
DSM-CBF, but ultimately settling farther from the constraints than the
Lyapunov RDSM-RCBF. To examine exactly why the trajectory RDSM-RCBF is
settling further from the constraints than the Lyapunov RDSM-RCBF despite both
approaches using the same weight factor $P$, refer to
Fig.~\ref{fig:dint-tube}.

\begin{figure}[ht]
\centering
  \includegraphics[width=0.90\linewidth]{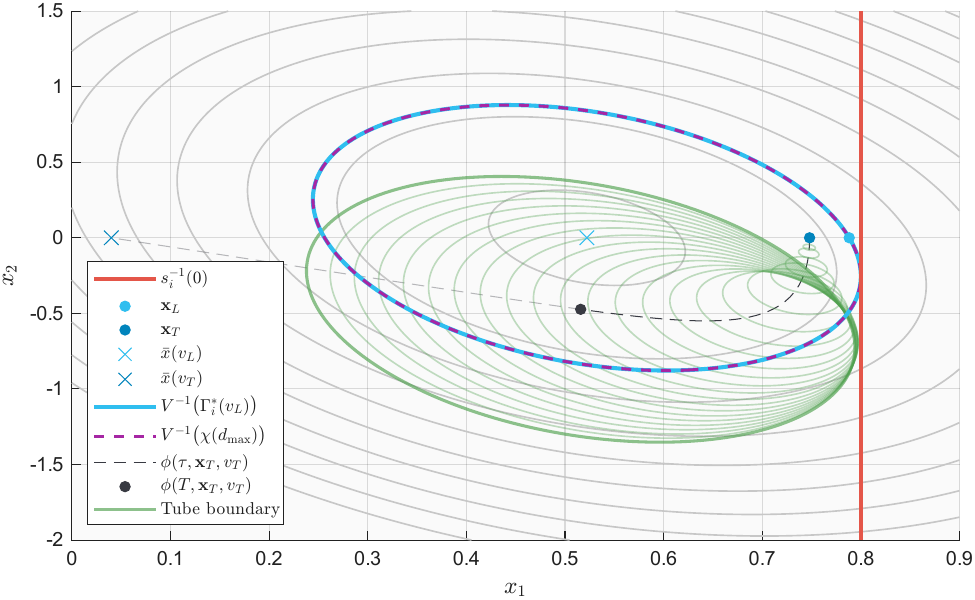}
  \caption{Steady-state analysis of Fig.~\ref{fig:dint} with
  respect to constraint $s_i(\bm{x}, v) = 0.8~\unit{m} - x$. The
  points $(\bm{x}_L, v_L)$ and $(\bm{x}_T, v_T)$ are the settling
  state-reference points of the Lyapunov and trajectory RDSM-RCBF approaches.
  For the Lyapunov RDSM-RCBF, we can see the approach advanced until
  $\Gamma^*(v_L) = \chi(d_{\max})$. Further advancement would cause the
  robust invariant level set of the ISS Lyapunov function described by
  $\chi(d_{\max})$ to breach the constraints. For the trajectory RDSM-RCBF, we
  can see the approach advanced until the boundary of the contraction tube
  $\sR$ given in \eqref{eq:contraction-tube} grazes the constraints
  at some $\tau \in [0,T]$ (not necessarily $\tau = T$), and any further
  advancement would cause the tube to breach the constraints. The settling
  point $\bm{x}_L$ is closer to the constraint boundary than $\bm{x}_T$ in
  this example due to the geometry of the contraction tube. This limitation
  can be addressed by obtaining tighter contraction tubes using, for example,
  non-Eucledian norms instead of $\|\cdot\|_2$. More details about contraction
  results for non-Euclidean norms abound in \cite{bullo2026contraction}.}
  \label{fig:dint-tube}
\end{figure}

\subsection{Inverted Pendulum}
Consider an inverted pendulum system subject to disturbances
\begin{equation}
    \ddot{\theta} = \frac{g_a}{L} \sin\theta + u + d,
\end{equation}
where $g_a = 10~\unit{m/s^2}$ is the acceleration of gravity and $L=1~\unit{m}$ is the length of the pendulum.
The equilibrium manifold is described by
\begin{equation*}
    \bar{x}(v) = \mat{v \\ 0}, \quad \bar{u}(v) = \frac{g_a}{L}\sin v.
\end{equation*}
The state and input constraint sets are
\begin{align*}
  \sX &= \{\bm{x} \in \R^2\mid |\theta| \leq \theta_{\max} = 12~\unit{\deg}\}, \\
  \sU &= \lbrack -u_{\max}, u_{\max} \rbrack = \lbrack -600~\unit{\deg/s^2}, 600~\unit{\deg/s^2} \rbrack,
\end{align*}
and the disturbance set is $\sW = \lbrack -d_{\max}, d_{\max}\rbrack$ with $d_{\max} = 22.9~\unit{deg/s^2}$. The initial state is $\bm{x}_0 = 0 \in \sX$ and the nominal controller is
\begin{equation*}
    \kappa(t,\bm{x}) = \bar{u}\big(r(t)\big) - K_{\kappa}\Big(\bm{x} - \bar{x}\big(r(t)\big)\Big),
\end{equation*}
with gain $K = \lbrack 60 \;\, 4\rbrack$, and where $r(t)$ is a periodic square wave reference signal alternating between $10~\unit{\deg}$ and $-10~\unit{\deg}$.

Considering prestabilizing controller
\begin{equation*}
    \pi(\bm{x}, v) = \bar{u}(v) -K\big(\bm{x}, \bar{x}(v)\big),
\end{equation*}
and defining the error vector $\bm{e} = \bm{x} - \bar{x}(v)$, we have that the prestabilized error dynamics satisfy $\dot{\bm{e}} = A_{\pi}(\xi) \bm{e} + Ed$, where
\begin{equation*}
    A_{\pi}(\xi) = \mat{0 & 1 \\ \frac{g_a}{L} \cos \xi - k_1 & k_2}, \; K = \mat{k_1 & k_2}, \; E = \mat{0 \\ 1}.
\end{equation*}
As shown in \cite{jafarpour2025probabilistic}, since $\cos(\xi)\in \lbrack-1, 1\rbrack$, $A_{\pi}(\xi)$ lies in the convex hull of two extreme vertex matrices
\begin{equation*}
  A_1 = \mat{0 & 1 \\ \frac{g_a}{L}-k_1 & -k_2}, \quad A_2 = \mat{0 & 1 \\ -\frac{g_a}{L} - k_1 & -k_2}.
\end{equation*}
Consider the prestabilized constraint function
\begin{equation*}
  s(\bm{x}, v) = \mat{ \theta + \theta_{\max} \\ \theta_{\max} - \theta \\ \pi(\bm{x}, v) + u_{\max} \\ u_{\max} - \pi(\bm{x}, v) },
\end{equation*}
which satisfies $\sX_v = \{\bm{x} \in \R^2 \mid s(\bm{x}, v) \geq 0\}$.

If $k_1 > g/L$ and $k_2 > 0$, then for any $Q \succ 0$, there exists $P \succ 0$ that satisfies both
\begin{align*}
  &PA_1 + A_1^\top P \prec -Q,\\
  &PA_2 + A_2^\top P \prec -Q.
\end{align*}
Specifically, we pick
\begin{equation*}
  P = \mat{1.72 & 0.24 \\ 0.24 & 0.12}, \quad Q = \mat{4.32 & 0.61 \\ 0.61 & 0.31}.
\end{equation*}
With this, it can be shown that the following is a reference-dependent ISS Lyapunov function satisfying Definition~\ref{def:lyap-iss}
\begin{equation*}
  V(\bm{x}, v) = \big(\bm{x} - \bar{x}(v)\big)^\top P \big(\bm{x} - \bar{x}(v)\big),
\end{equation*}
with $\chi$ and $\Gamma^*$ given as in Section~\ref{ssec:lyap-rdsm}. Moreover, it can also be shown that Assumptions~\ref{ass:contractivity}--\ref{ass:bounded-disturbance-function} are satisfied with
\begin{equation*}
    c = \frac{1}{2}\lambda_{\min}\left(P^{-1/2} Q P^{-1/2}\right) > 0, \quad c_d = \left\|P^{1/2} E\right\|> 0.
\end{equation*}
Using these derivations, we compare the performance of the trajectory DSM-CBF, the Lyapunov RDSM-RCBF, and the trajectory RDSM-RCBF in Fig.~\ref{fig:pend}. The prediction horizon for all approaches that use it is $T=0.5~\unit{s}$ and all class $\cK$ functions used are $\alpha:c \mapsto 50c$.
Both robust approaches achieve safety without the need for slack variables in the RDSM-RCBF safety filter. Moreover, the trajectory RDSM-RCBF exhibits faster rise times than the Lyapunov RDSM-RCBF. The non-robust trajectory DSM-CBF performs well initially, but at around time 6.5, the disturbances cause a spike in the slack variables that derails the simulation and results in safety violations.

\begin{figure}[t]
\centering
  \includegraphics[width=0.90\linewidth]{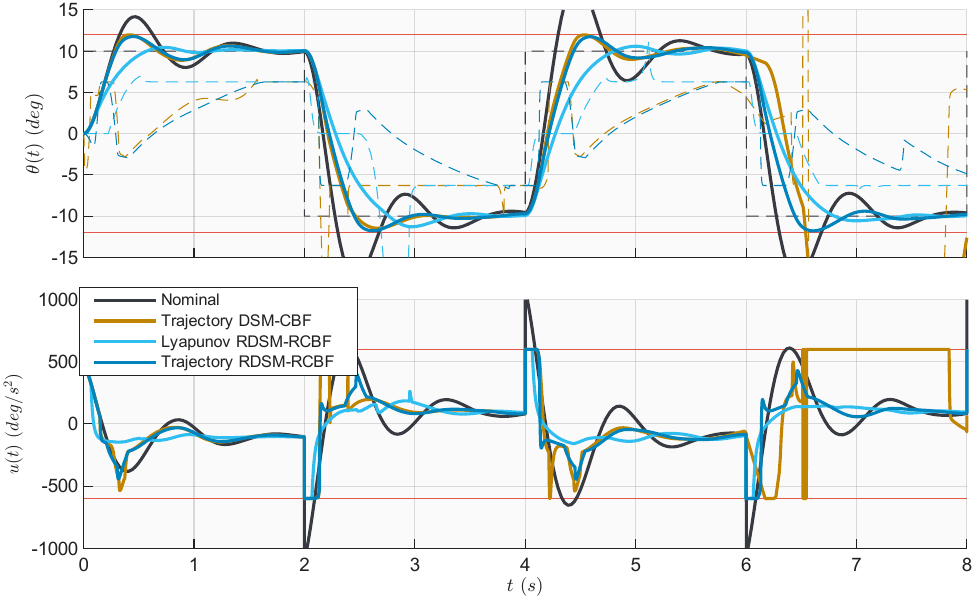}
  \includegraphics[width=0.90\linewidth]{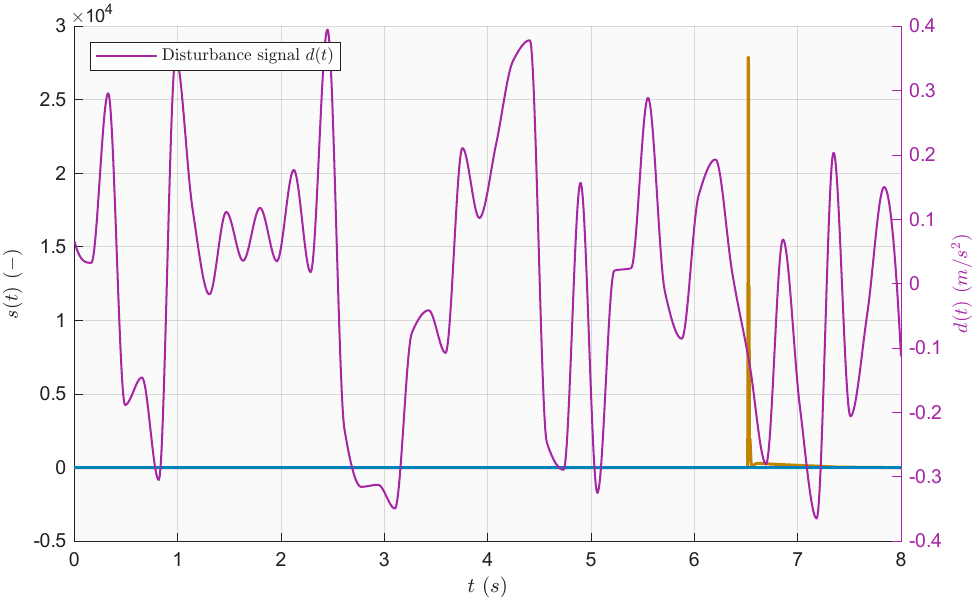}
  \caption{Simulation of the inverted pendulum model subject to disturbances. Both robust approaches successfully enforce all the constraints while the trajectory DSM-CBF fails to do so. In addition, the trajectory DSM-CBF required augmenting the DSM-CBF safety filter with a slack variable to prevent infeasibility, its value is given in the bottom plot, which also shows that the robust approaches did not require the use of their slack variable to maintain feasibility as expected.}
  \label{fig:pend}
\end{figure}

\section{Conclusion} \label{sec:conclusion}
In this paper, we showed that robust dynamic safety margins are robust control barrier functions for an augmented system consisting of the concatenation of the state and a virtual reference. In addition, we formulated an augmented the RDSM-RCBF safety filter to distill safe controllers. Next, we presented two systematic approaches to construct RDSM-RCBFs leveraging, respectively, ISS Lyapunov functions and contractive trajectory predictions. Moreover, provided closed-form derivations for the special case of linear systems. Finally, we highlighted the robustness properties of the presented approaches against their nonrobust counterparts in two simulation examples.

\bibliographystyle{IEEEtran}
\bibliography{references}

@inproceedings{nagumo1942lage,
  title={{\"U}ber die lage der integralkurven gew{\"o}hnlicher differentialgleichungen},
  author={Nagumo, Mitio},
  booktitle={Proc. Phys.-Math. Soc. Jpn.},
  volume={24},
  pages={551--559},
  year={1942},
}

@book{berge1963topological,
  title={Topological spaces: Including a treatment of multi-valued functions, vector spaces and convexity},
  author={Berge, Claude},
  year={1963},
  publisher={Oliver \& Boyd}
}

@article{mangasarian1967fritz,
  title={The {F}ritz {J}ohn necessary optimality conditions in the presence of equality and inequality constraints},
  author={Mangasarian, Olvi L and Fromovitz, Stan},
  journal={J. Math. Anal. Appl.},
  volume={17},
  number={1},
  pages={37--47},
  year={1967},
  publisher={Elsevier}
}

@article{sontag1995characterizations,
  title={On characterizations of the input-to-state stability property},
  author={Sontag, Eduardo D and Wang, Yuan},
  journal={Systems \& Control Letters},
  volume={24},
  number={5},
  pages={351--359},
  year={1995},
  publisher={Elsevier}
}

@article{albertini1999continuous,
  title={Continuous control-Lyapunov functions for asymptotically controllable
         time-varying systems},
  author={Albertini, Francesca and Sontag, Eduardo D},
  journal={International Journal of Control},
  volume={72},
  number={18},
  pages={1630--1641},
  year={1999},
  publisher={Taylor \& Francis}
}

@book{khalil2002nonlinear,
  author = "Khalil, Hassan K.",
  title = "{Nonlinear systems}",
  publisher = "Prentice-Hall",
  address = "Upper Saddle River, NJ",
  year = "2002",
}

@book{blanchini2008set,
  title={Set-theoretic methods in control},
  author={Blanchini, Franco and Miani, Stefano},
  year={2008},
  publisher={Springer}
}

@book{aubin2009set,
  title={Set-valued analysis},
  author={Aubin, Jean-Pierre and Frankowska, H{\'e}l{\`e}ne},
  year={2009},
  publisher={Springer},
  note={Reprint of the 1990 edition}
}

@inproceedings{nicotra2016explicit,
  title={An explicit reference governor for the robust constrained control of nonlinear systems},
  author={Nicotra, Marco M and Garone, Emanuele},
  booktitle={2016 IEEE 55th Conference on Decision and Control (CDC)},
  pages={1502--1507},
  year={2016},
  organization={IEEE}
}

@book{bertsekas2016nonlinear,
  edition = {3rd edition},
  title = {Nonlinear programming},
  isbn = {978-1-886529-05-2},
  publisher = {Athena Scientific},
  author = {Bertsekas, Dimitri P.},
  year = {2016},
}

@article{garone2017reference,
  title={Reference and command governors for systems with constraints: A survey on theory and applications},
  author={Garone, Emanuele and Di Cairano, Stefano and Kolmanovsky, Ilya},
  journal={Automatica},
  volume={75},
  pages={306--328},
  year={2017},
  publisher={Elsevier}
}

@article{xu2018constrained,
  title={Constrained control of input--output linearizable systems using control sharing barrier functions},
  author={Xu, Xiangru},
  journal={Automatica},
  volume={87},
  pages={195--201},
  month={Jan.},
  year={2018},
  publisher={Elsevier}
}

@article{van2024disturbance,
  title={Disturbance-robust backup control barrier functions: Safety under uncertain dynamics},
  author={van Wijk, David EJ and Coogan, Samuel and Molnar, Tamas G and Majji, Manoranjan and Hobbs, Kerianne L},
  journal={IEEE Control Systems Letters},
  volume={8},
  pages={2817--2822},
  year={2024},
  publisher={IEEE}
}

@Book{bullo2026contraction,
  author =    {F. Bullo},
  title =     {Contraction Theory for Dynamical Systems},
  year =      2026,
  publisher = {Kindle Direct Publishing},
  ISBN =      {979-8836646806},
  url =       {https://fbullo.github.io/ctds},
}

@article{dacs2026robust,
  title={Robust Adaptive Backup Control Barrier Functions},
  author={Da{\c{s}}, Ersin and van Wijk, David EJ and Molnar, Tamas G and Ames, Aaron D and Burdick, Joel W},
  journal={arXiv preprint arXiv:2607.20842},
  year={2026}
}

@inproceedings{choi2021robust,
  title = {Robust control barrier--value functions for safety-critical control},
  author = {Choi, Jason J. and Lee, Donggun and Sreenath, Koushil and Tomlin, Claire J. and Herbert, Sylvia L.},
  booktitle = {Proc. 60th IEEE Conf. Decis. Control},
  pages = {6814--6821},
  month = {December},
  year = {2021},
  address = {Austin, TX, USA},
}

@article{wabersich2023data,
  title = {Data-driven safety filters: {H}amilton--{J}acobi reachability, control barrier functions, and predictive methods for uncertain systems},
  author = {Wabersich, Kim P. and Taylor, Andrew J. and Choi, Jason J. and Sreenath, Koushil and Tomlin, Claire J. and Ames, Aaron D. and Zeilinger, Melanie N.},
  journal = {IEEE Control Syst. Mag.},
  volume = {43},
  number = {5},
  pages = {137--177},
  month = {October},
  year = {2023},
}

@inproceedings{wang2018permissive,
  title = {Permissive barrier certificates for safe stabilization using sum-of-squares},
  author = {Wang, Li and Han, Dongkun and Egerstedt, Magnus},
  booktitle = {Proc. Amer. Control Conf.},
  pages = {585--590},
  month = {June},
  year = {2018},
  address = {Milwaukee, WI, USA},
}

@article{clark2025semialgebraic,
  title = {A semialgebraic framework for verification and synthesis of control barrier functions},
  author = {Clark, Andrew},
  journal = {IEEE Trans. Autom. Control},
  volume = {70},
  number = {5},
  pages = {3101--3116},
  month = {May},
  year = {2025},
}

@inproceedings{srinivasan2020synthesis,
  title = {Synthesis of control barrier functions using a supervised machine learning approach},
  author = {Srinivasan, Mohit and Dabholkar, Amogh and Coogan, Samuel and Vela, Patricio A.},
  booktitle = {Proc. IEEE/RSJ Int. Conf. Intell. Robots Syst.},
  pages = {7139--7145},
  month = {October},
  year = {2020},
  address = {Las Vegas, NV, USA},
}

@inproceedings{robey2020learning,
  title = {Learning control barrier functions from expert demonstrations},
  author = {Robey, Alexander and Hu, Haimin and Lindemann, Lars and Zhang, Hanwen and Dimarogonas, Dimos V. and Tu, Stephen and Matni, Nikolai},
  booktitle = {Proc. 59th IEEE Conf. Decis. Control},
  pages = {3717--3724},
  month = {December},
  year = {2020},
  address = {Jeju, Republic of Korea},
}

@article{bajelani2025data,
  title = {Data-driven input-output control barrier functions},
  author = {Bajelani, Mohammad and van Heusden, Klaske},
  journal = {IEEE Control Syst. Lett.},
  volume = {9},
  pages = {216--221},
  month = {May},
  year = {2025},
}

@article{kaviani2026synthesis,
  title = {Synthesis of Data-Driven Discrete-Time Control Barrier Functions Using Maximal Admissible Sets},
  author = {Kaviani, Farzan and Pedari, Yasaman and Ossareh, Hamid R.},
  journal = {IEEE Control Syst. Lett.},
  volume = {10},
  pages = {907--912},
  month = {June},
  year = {2026},
}

@article{gurriet2020scalable,
  title = {A scalable safety critical control framework for nonlinear systems},
  author = {Gurriet, Thomas and Mote, Mark and Singletary, Andrew and Nilsson, Petter and Feron, Eric and Ames, Aaron D.},
  journal = {IEEE Access},
  volume = {8},
  pages = {187249--187275},
  month = {September},
  year = {2020},
}

@inproceedings{chen2021backup,
  title = {Backup control barrier functions: Formulation and comparative study},
  author = {Chen, Yuxiao and Jankovic, Mrdjan and Santillo, Mario and Ames, Aaron D.},
  booktitle = {Proc. 60th IEEE Conf. Decis. Control},
  pages = {6835--6841},
  month = {December},
  year = {2021},
  address = {Austin, TX, USA},
}

@article{van2026generalizations,
  title = {Generalizations of Backup Control Barrier Functions: Expansion and Adaptation for Input-Bounded Safety-Critical Control},
  author = {van Wijk, David E.J. and Lee, Dohyun and Da{\c{s}}, Ersin and Molnar, Tamas G. and Ames, Aaron D. and Burdick, Joel W.},
  journal = {IEEE Control Syst. Lett.},
  volume = {10},
  pages = {1945--1950},
  month = {July},
  year = {2026},
}

@inproceedings{van2026safety,
  title = {Safety-critical control with bounded inputs: A closed-form solution for backup control barrier functions},
  author = {van Wijk, David E.J. and Das, Ersin and Molnar, Tamas G. and Ames, Aaron D. and Burdick, Joel W.},
  booktitle = {Proc. Amer. Control Conf.},
  pages = {},
  note = {to appear},
  month = {May},
  year = {2026},
  address = {New Orleans, LA, USA},
}

@article{jankovic2018robust,
  title = {Robust control barrier functions for constrained stabilization of nonlinear systems},
  author = {Jankovic, Mrdjan},
  journal = {Automatica},
  volume = {96},
  pages = {359--367},
  month = {October},
  year = {2018},
}

@article{kolathaya2019input,
  title = {Input-to-state safety with control barrier functions},
  author = {Kolathaya, Shishir and Ames, Aaron D.},
  journal = {IEEE Control Syst. Lett.},
  volume = {3},
  number = {1},
  pages = {108--113},
  month = {January},
  year = {2019},
}

@article{alan2021safe,
  title = {Safe controller synthesis with tunable input-to-state safe control barrier functions},
  author = {Alan, Anil and Taylor, Andrew J. and He, Chaozhe R. and Orosz, G{\'a}bor and Ames, Aaron D.},
  journal = {IEEE Control Syst. Lett.},
  volume = {6},
  pages = {908--913},
  month = {June},
  year = {2021},
}

@article{nicotra2018explicit,
  title = {The explicit reference governor: A general framework for the closed-form control of constrained nonlinear systems},
  author = {Nicotra, Marco M. and Garone, Emanuele},
  journal = {IEEE Control Syst. Mag.},
  volume = {38},
  number = {4},
  pages = {89--107},
  month = {August},
  year = {2018},
}

@article{freire2023systematic,
  title = {Systematic design of discrete-time control barrier functions using maximal output admissible sets},
  author = {Freire, Victor and Nicotra, Marco M.},
  journal = {IEEE Control Syst. Lett.},
  volume = {7},
  pages = {1891--1896},
  month = {June},
  year = {2023},
  doi = {10.1109/LCSYS.2023.3282888},
}

@article{li2023governor,
  title={Governor-parameterized barrier function for safe output tracking with locally sensed constraints},
  author={Li, Zhichao and Atanasov, Nikolay},
  journal={Automatica},
  volume={152},
  pages={110996},
  month={June},
  year={2023},
}

@book{aubin1991viability,
  title = {Viability Theory},
  author = {Aubin, Jean-Pierre},
  year = {1991},
  publisher = {Birkh{\"a}user},
  address = {Boston, MA, USA},
}

@article{blanchini1999set,
  title = {Set invariance in control},
  author = {Blanchini, Franco},
  journal = {Automatica},
  volume = {35},
  number = {11},
  pages = {1747--1767},
  month = {November},
  year = {1999},
}

@article{jafarpour2025probabilistic,
  title={Probabilistic reachability analysis of stochastic control systems},
  author={Jafarpour, Saber and Liu, Zishun and Chen, Yongxin},
  journal={IEEE Trans. Autom. Control},
  volume={70},
  number={11},
  pages={7080--7094},
  year={2025},
  month={November},
}

@manual{mosek,
  author = {MOSEK~ApS},
  title = {The MOSEK optimization toolbox for MATLAB manual. Version 11.2.2.},
  year = {2026},
}

@inproceedings{freire2026designing,
  author = {Freire, Victor and Nicotra, Marco M.},
  title = {Designing Control Barrier Functions Using a Dynamic Backup Policy},
  booktitle = {Proc. 23rd IFAC World Congr.},
  year = {2026},
  month = {August},
  address = {Busan, Republic of Korea},
}

@article{freire2026using,
  author = {Freire, Victor and Nicotra, Marco M.},
  title = {Using dynamic safety margins as control barrier functions},
  journal = {IEEE Trans. Autom. Control},
  year = {2026},
  month = {September},
  doi = {10.1109/TAC.2026.3686634},
  volume = {71},
  issue = {9},
  pages = {5908--5919},
}

\end{document}